\documentclass[letterpaper, 11pt]{article}
\usepackage[margin=1in]{geometry}
\usepackage{thm-restate}
\usepackage{authblk}
\usepackage[T1]{fontenc}
\usepackage{soul} 
\usepackage{phfqit}
\usepackage{proba} 
\usepackage{amssymb}
\usepackage{amsmath}
\usepackage{amsthm}
\usepackage{amsfonts}
\usepackage{mathtools}
\usepackage{xspace}
\usepackage{bm}
\usepackage{nicefrac}
\usepackage{commath}
\usepackage{bbm}
\usepackage{boxedminipage}
\usepackage{xparse}
\usepackage{xcolor}
\usepackage{float}
\usepackage{multirow}
\usepackage{graphicx}
\usepackage{caption}
\usepackage{subcaption}
\usepackage{ifthen}
\usepackage{algorithm}
\usepackage{algpseudocode}
\usepackage{colortbl} 
\usepackage{fancyhdr}   
\usepackage{hyperref}
    \hypersetup{ colorlinks=true, linkcolor=blue, filecolor=magenta, urlcolor=red,}
\usepackage{fullpage}
\usepackage[utf8]{inputenc}
\usepackage{environ}
\usepackage[colorinlistoftodos]{todonotes}
\usepackage{mdframed}

\usepackage{tikz}
\usetikzlibrary{shapes}
\usetikzlibrary{positioning}
\usetikzlibrary{fit}
\usetikzlibrary{calc}
\usetikzlibrary{backgrounds}
\usetikzlibrary{patterns}
\usetikzlibrary{matrix}
\usetikzlibrary{arrows}

\newtheorem{proposition}{Proposition}
\newtheorem{lemma}{Lemma}
\newtheorem{theorem}{Theorem}
\newtheorem{corollary}{Corollary}
\newtheorem{definition}{Definition}

\newtheorem{claim}{Claim}

\newcommand{\polylog}{\ensuremath{\mathrm{polylog}}\xspace}
\newcommand{\negl}{\ensuremath{\mathsf{negl}}\xspace}

\newcommand{\tuple}[1]{\ensuremath{\left(#1\right)}}
\newcommand{\tp}{\tuple}
\newcommand{\eps}{\ensuremath{\varepsilon}}

\newcommand{\ceil}[1]{\ensuremath{\left\lceil{#1}\right\rceil}\xspace}
\newcommand{\floor}[1]{\ensuremath{\left\lfloor{#1}\right\rfloor}\xspace}

\newcommand{\itn}[1]{^{(#1)}}
\def\*#1{\mathbf{#1}}
\def\+#1{\mathcal{#1}}

\NewEnviron{problem}[1]{%
	\begin{center}\fbox{\parbox{6in}{%
				{\centering\scshape #1\par}%
				\parskip=1ex
				\everypar{\hangindent=1em}%
				\BODY
}}\end{center}}

    \newcommand{\Enc}{\ensuremath{\mathsf{Enc}}\xspace}
    \newcommand{\Dec}{\ensuremath{\mathsf{Dec}}\xspace}

    \newcommand{\ED}{\ensuremath{\mathsf{ED}}\xspace}

    \newcommand{\ignore}[1]{{}}

\newcommand{\Goody}{\ensuremath{\mathsf{Good}}}

\def\final{0} 
\ifnum\final=0  
\newcommand{\authnote}[3]{\textcolor{#2}{{\sf (#1's Note: {\sl{#3}})}}}
\newcommand{\jeremiah}{\authnote{Jeremiah}{blue}}
\newcommand{\minshen}{\authnote{Minshen}{purple}}
\newcommand{\xnote}{\authnote{Xin}{magenta}}
\newcommand{\enote}{\authnote{Elena}{green}}
\newcommand{\knote}{\authnote{Kuan}{cyan}}
\newcommand{\yu}{\authnote{Yu}{cyan}}
\else 
\newcommand{\authnote}{}
\newcommand{\jeremiah}[1]{}
\newcommand{\minshen}[1]{}
\newcommand{\xnote}[1]{}
\newcommand{\enote}[1]{}
\newcommand{\knote}[1]{}
\newcommand{\yu}[1]{}
\fi

\newif\ifnotes\notestrue 

\begin{document}

\title{Exponential Lower Bounds for Locally Decodable and Correctable Codes for Insertions and Deletions}

\author[1]{Jeremiah Blocki\thanks{Supported by NSF CAREER Award CNS-2047272 and NSF Awards CCF-1910659 and CNS-1931443,}}
\author[2,3]{Kuan Cheng\thanks{Supported by the National Natural Science Foundation of China under Grant 62472008 and CCF-Huawei Populus Grove Fund CCF-HuaweiLK2025005.}}
\author[4]{Elena Grigorescu\thanks{Supported by NSF CCF-1910659 and NSF CCF-1910411 while at Purdue University.}}
\author[5]{Xin Li \thanks{Supported by NSF CAREER Award CCF-1845349 and NSF Award CCF-2127575.}}
\author[5]{Yu Zheng \thanks{Supported by NSF CAREER Award CCF-1845349.}}
\author[1]{Minshen Zhu}
\affil[1]{Department of Computer Science, Purdue University}
\affil[2]{Center on Frontiers of Computing Studies, Peking University}
\affil[3]{Advanced Institute of Information Technology, Peking University}
\affil[4]{Cheriton School of Computer Science, University of Waterloo}
\affil[5]{Department of Computer Science, Johns Hopkins University}
\affil[ ]{\textit {\{jblocki, zhu628\}@purdue.edu}}
\affil[ ]{\textit {ckkcdh@pku.edu.cn} }
\affil[ ]{\textit {elena-g@uwaterloo.ca}}
\affil[ ]{\textit {\{lixints, yuzheng\}@cs.jhu.edu} }

\date{}

\maketitle
\thispagestyle{empty}
\vspace{-0.5in}
\begin{abstract}
 Locally Decodable Codes (LDCs) are error-correcting codes for which individual message symbols can be quickly recovered despite errors in the codeword. LDCs for Hamming errors have been studied extensively in the past few decades, where a major goal is to understand the amount of redundancy that is necessary and sufficient to decode from large amounts of error, with small query complexity.
 
Motivated by new progress in DNA-storage technologies (Banal et al., {\em Nature Materials, 2021}), in this work we study LDCs for {\em insertion} and {\em deletion} errors, called {\em Insdel LDCs}. Their study was initiated by Ostrovsky and Paskin-Cherniavsky ({\em Information Theoretic Security, 2015}), who gave a reduction from Hamming LDCs to Insdel LDCs with a small blowup in the code parameters. On the other hand, the only known lower bounds for Insdel LDCs come from those for Hamming LDCs, thus there is no separation between them. Here we prove new, strong lower bounds for the existence of Insdel LDCs. In particular, we show that $2$-query {\em linear} Insdel LDCs do not exist, and give an exponential lower bound for the length of all $q$-query Insdel LDCs with constant $q$. For $q \ge 3$ our bounds are exponential in the existing lower bounds for Hamming LDCs. Furthermore, our exponential lower bounds continue to hold for adaptive decoders, and even in private-key settings where the encoder and decoder share secret randomness. This exhibits a strict separation between Hamming LDCs and Insdel LDCs.
 
 Our strong lower bounds also hold for the related notion of Insdel LCCs (except in the private-key setting), due to an  analogue to the Insdel notions of a reduction from Hamming LCCs to LDCs.
 
 Our techniques are based on a delicate design and analysis of hard distributions of insertion and deletion errors, which depart significantly from typical techniques used in analyzing Hamming LDCs.

\end{abstract}

\pagenumbering{arabic} 
\section{Introduction}
Error correcting codes are fundamental mathematical objects in both theory and practice, whose study dates back to the pioneering work of Shannon and Hamming in the 1950's. While  the study of classical codes focuses on unique decoding from {\em Hamming} errors, many exciting variants have emerged ever since, such as list-decoding, which can go beyond the half distance barrier, and {\em local decoding},  which can decode any message symbol by querying only a few codeword symbols. These variants have proved to be closely connected to diverse areas in computer science. 

Similarly, another line of work studies {\em synchronization} errors, namely {\em insertions} and {\em deletions} ({\em insdels}, for short), which are strictly more general than Hamming errors and happen frequently in various applications such as text/speech processing, media access, and communication systems. The study of codes for such errors ({\em insdel codes}, for short) also has a long history that goes back to the work of Levenstein \cite{Levenshtein_SPD66} in the 1960's.

This paper focuses on locally decodable codes correcting insertions and deletions, which we call {\em Insdel LDCs}. We prove the first non-trivial lower bounds for such codes, which in turn provide a strong separation between Hamming LDCs and Insdel LDCs. Furthermore, these results imply similar strong bounds for the related notion of locally correctable codes correcting insertions and deletions, which we call {\em Insdel LCCs}.

More formally, Locally Decodable Codes (LDCs) are error-correcting codes $C: \Sigma^n \rightarrow \Sigma^m$ that allow very fast recovery of individual symbols of a {\em message} $x\in \Sigma^n$, even when worst-case errors are introduced in the encoded message, called {\em codeword} $C(x)$. Similarly, Locally Correctable Codes (LCCs) are error-correcting codes $C: \Sigma^n \rightarrow \Sigma^m$ that allow very fast recovery of individual symbols of the codeword $C(x)\in \Sigma^m$, even when worst-case errors are introduced. In what follows, for ease of presentation, we will discuss our results and related work by focusing on the notion of LDCs, and we will return to the implications to Insdel LCCs in Section \ref{sec:intro-LCCs}. We remark that the previous lower bounds for Hamming LCCs are asymptotically the same as for LCCs due to a folklore reduction between the two notions (e.g. formalized in \cite{KaufmanV10,bhattacharyya2016lower}).

The important {\em parameters} of LDCs are their {\em  rate}, defined as the ratio between the message length $n$ and the codeword length $m$, measuring the amount of redundancy in the encoding; their {\em relative minimum distance}, defined as the minimum normalized Hamming distance between any pair of codewords, a parameter relevant to the fraction of correctable errors;  and their {\em locality} or {\em query complexity}, defined as the number of queries a decoder makes to a received word  $y\in \Sigma^m$ in order to decode the symbol at location $i\in [n]$ of the message, namely $x_i$.

Since they were introduced in \cite{KatzT00, SudanTV99}, LDCs have found many applications in private information retrieval, probabilistically checkable proofs, self-correction,  fault-tolerant circuits,  hardness amplification, and data structures (e.g., \cite{BabaiFLS91,LundFKN92,BlumLR93,BlumK95,ChorKGS98,ChenGW13,AndoniLRW17} and surveys \cite{Tre04-survey,Gasarch04}), and  
the tradeoffs between the achievable parameters of Hamming LDCs has been studied extensively \cite{KerenidisW04, WehnerW05, GoldreichKST06, Woodruff07, Yekhanin08, Yekhanin12, DvirGY11, Efremenko12,GalM12, BhattacharyyaDS16, BhattacharyyaG17,DvirSW17,KoppartyMRS17, BhattacharyyaCG20} (see also surveys by Yekhanin \cite{Yekhanin12} and by Kopparty and Saraf \cite{KoppartyS16}). This sequence of results has brought up exciting progress regarding the necessary and sufficient rate for codes with small query complexity that can withstand a constant fraction of errors.  Nevertheless, many important parameter regimes leave wide gaps in our current understanding of LDCs. For example, even for $3$-query Hamming LDCs the gap between constructions and lower bounds is superpolynomial \cite{Yekhanin08,DvirGY11,Efremenko12, KatzT00,Woodruff07,Woodruff12}. (Note:  \cite{GalM12} established an exponential lower bound on the length of $3$-query LDCs for some parameter regimes, but it does not rule out the possibility of a $3$-query LDC with polynomial length in natural parameter ranges.) 

More specifically, for $2$-query Hamming LDCs we have matching upper and lower bounds of $m=2^{\Theta(n)}$, where the upper bound is achieved by the simple Hadamard code while the lower bound is established in \cite{KerenidisW04, Ben-AroyaRW08}. In the constant-query regime where the decoder makes $2^t$ many queries, for some $t>1$, the best known constructions of Hamming LDCs are based on matching-vector codes, and give codes that map $n$ symbols into $m=\exp( \exp((\log n)^{1/t} (\log\log n)^{1-1/t}))$  
symbols~\cite{Yekhanin08,DvirGY11,Efremenko12},
while the best general lower bound for $q$-query LDC is $\Omega(n^{\frac{q+1}{q-1}})/\log n$ when $q \ge 3$ \cite{Woodruff07}. 
In the $\polylog(n)$-query regime, Reed-Muller codes are examples of $\log^c n$-query Hamming LDCs of block length $n^{1+\frac{1}{c-1}+o(1)}$ for some $c>0$ (e.g., see \cite{Yekhanin12}). Finally, there exist sub-polynomial (but super logarithmic)-query Hamming LDCs with constant rate \cite{KoppartyMRS17}. These latter constructions  improve upon the previous constant-rate codes in the  $n^{\epsilon}$-query regime achieved by Reed-Muller codes, and upon the more efficient constructions of  \cite{KoppartySY14}. In a different model, if we assume that the encoder and decoder have shared secret randomness \cite{Lipton94}, then it becomes much easier to construct LDCs. For example, \cite{OPS07} constructs private-key Hamming LDCs with constant rate (i.e., $m= \Theta(n)$) and query complexity $\polylog(n),$ and a simple modification yields a private-key Hamming LDC with rate $m = \tilde{\Theta}(n)$ and query complexity $1$ --- see details in Appendix \ref{sec:note}.

 Regarding insdel codes, following the work of Levenstein \cite{Levenshtein_SPD66}, the progress has historically been slow, due to the fact that synchronization errors often result in the loss of index information. Indeed, constructing codes for insdel errors is strictly more challenging than for Hamming errors. However the interest in these codes has been rekindled lately, leading to a wave of new results  \cite{SchZuc99, Kiwi_expectedlength, guruswami2017deletion,HaeuplerS17, GuruswamiL18,HaeuplerSS18,HaeuplerS18,BrakensiekGZ18, ChengJLW18, ChengHLSW19, ChengJ0W19, GuruswamiL19,HaeuplerRS19,Haeupler19,  LiuTX20,GuruswamiHS20, ChengGHL21, ChengL21} (See also the excellent surveys of \cite{Sloane2002OnSC,Mercier2010ASO,Mitzenmachen-survey, haeupler2021synchronization}) with almost optimal parameters in various settings, and the variant of ``list-decodable'' insdel codes, that can withstand a larger fraction of errors while outputting a small list of potential codewords  \cite{HaeuplerSS18,GuruswamiHS20,LiuTX20}. However, none of these works addresses insdel LDCs, which we believe are natural objects in the study of insdel codes, since such codes are often used in applications involving large data sets.

Insdel LDCs were first introduced in \cite{Ostrovsky-InsdelLDC-Compiler} and further studied in \cite{BlockBGKZ20,block2021private,ChengLZ20}.  In \cite{Ostrovsky-InsdelLDC-Compiler,BlockBGKZ20} the authors give Hamming to Insdel reductions which transform any Hamming LDC into an Insdel LDC. These reductions decrease the rate by a constant multiplicative factor and increase the locality by a  $\log^{c'}(m)$ multiplicative factor for a fixed constant $c'>1$. Applying the compilers to the above-mentioned constructions of Reed-Muller codes gives $(\log n)^{c+c'}$-query Insdel LDCs of length $m=n^{1+\frac{1}{c-1}+o(1)}$, for any $c>1$. Also, applying the compilers to the LDCs in \cite{KoppartyMRS17} yields Insdel LDCs of constant rate and $\exp(\tilde{O}(\sqrt{\log n}))$-query complexity. 
  
Unfortunately, these compilers do not imply constant-query Insdel LDCs, and in fact, even after this work, we do not know if constant-query  Insdel LDCs exist in general. 
In the private-key setting, applying the compilers to~\cite{OPS07} yields a private-key Insdel LDC with constant rate and locality $\polylog(n)$~\cite{ChengLZ20, block2021private}. 
  
 We now formally define the notion of Insdel LDCs. See Appendix~\ref{sec:note} for  further discussion.  
 
\begin{restatable}{definition}{DefInsdelLDC}[Insdel Locally Decodable Codes (Insdel LDCs)] \label{def:InsdelLDC}
Fix an integer $q$ and constants $\delta \in [0,1]$, $\eps \in (0, \frac{1}{2}]$. We say $C \colon \set{0,1}^n \rightarrow \Sigma^m$ is a \emph{$(q,\delta,\eps)$-locally decodable insdel code} if there exists a probabilistic algorithm $\Dec$ such that:
\begin{itemize}
\item For every $x \in \set{0,1}^n$ and $y \in \Sigma^{m'}$ such that $\ED\tp{C(x), y} \le \delta \cdot 2m$, and for every $i \in [n]$, we have
\begin{align*}
\Pr\left[ \Dec(y, m', i) = x_i \right] \ge \frac{1}{2} + \eps,
\end{align*}
where the probability is taken over the randomness of $\Dec$, and $\ED\tp{C(x),y}$ denotes the minimum number of insertions/deletions necessary to transform $C(x)$ into $y$.
\item In every invocation, $\Dec$ reads at most $q$ symbols of $y$. We say that $\Dec$ is \emph{non-adaptive} if the distribution of queries of $\Dec(y, m', i)$ is independent of $y$.
\end{itemize}
\end{restatable}

Note that in this definition we allow the decoder to have as an input $m'$, the length of the string $y$. This only makes our lower bounds stronger. We can also extend the definition to private-key LDC, where the encoder and decoder share secret randomness, and we relax the requirement that $\Pr\left[ \Dec(y, m', i) = x_i \right] \ge \frac{1}{2} + \eps$ for all  $y$ s.t. $\ED\tp{C(x), y} \le \delta \cdot 2m$. Instead, we require that any attacker who does not have the secret randomness (private-key) cannot produce $y$ such that  $\ED\tp{C(x), y} \le \delta \cdot 2m$ and $\Pr\left[ \Dec(y, m', i) = x_i \right] < \frac{1}{2} + \eps$ except with negligible probability --- see Appendix \ref{subsec:privkeyinsdel}.

In this work we focus on {\em binary} Insdel LDCs and  give the first non-trivial lower bounds for such codes. In most cases, such as constant-query Insdel LDCs, our bounds are exponential in the message length. We note that prior to our work, the only known lower bounds for Insdel LDCs come from the lower bounds for Hamming LDCs (since Hamming erros can be implemented by insdel errors), and thus there is no separation of Insdel LDC and Hamming LDC. In particular, these bounds don't even preclude the possibility of a $3$-query Insdel LDC with $m=\Theta(n^2)$. We also note that we mainly prove lower bounds for Insdel LDCs with non-adaptive decoders. However, by using a reduction suggested in~\cite{KatzT00} we obtain almost the same lower bounds for Insdel LDCs with adaptive decoders.
 
Our results provide a strong separation between Insdel LDCs and Hamming LDCs in several contexts. First, many of our exponential lower bounds continue to apply in the setting of private-key LDCs, while in such settings it is easy to construct private-key Hamming LDCs with $m=\tilde{O}(n)$ and locality $1$. Second, our exponential lower bounds hold for any constant $q$, while even for $q = 3$ we have constructions of Hamming LDCs with sub-exponential length. Finally, for $q=2$ we rule out the possibility of linear Insdel LDCs, while the Hadamard code is a simple $2$-query Hamming LDC. This separation is in sharp contrast to the situation of unique decoding with codes for Hamming errors vs. codes for insdel errors, where they have almost the same parameter tradeoffs. 
 
 \paragraph{Motivation of Insdel LDCs in DNA storage} DNA storage \cite{Olgica17} is a storage medium that harnesses the biology of DNA sequences, to  store and transmit not only genetic information, but also any arbitrary digital data, despite the presence of insertion and deletion errors. It has the potential of  becoming the storage medium of the future, due to its superior scaling properties, provided new techniques for  random data access are developed. Recent progress towards achieving effective and reliable DNA  random access technology  is motivated by the fact that a {\em ``crucial aspect of data storage systems is the ability to efficiently retrieve specific files or arbitrary subsets of files.''} \cite{Banaletal-nature2021}. This is also precisely the real-world goal formalized by the notion of Insdel LDCs, which motivates a systematic theoretical study of such codes and of their limitations.

\subsection{Our results}

\subsubsection{Lower bounds for 2-query Insdel LDCs}

We first present our result for {\em linear} codes. Linear codes are defined over a finite field $\Sigma=\F$, and the codewords form a linear subspace in $\F^m$.  Similarly, the codewords of an {\em affine} code form an affine subspaces in $\F^m$. Lower bounds for the length of $2$-query linear Hamming LDCs were first studied in \cite{GoldreichKST06}, where the authors proved an exponential bound. This is matched by the Hadamard code.

In \cite{Woodruff12} Woodruff give a $m=\Omega(n^2)$ bound for $3$-query linear codes, which is still the best known for any linear code with $q \ge 3$. 

Furthermore, the best upper bounds of \cite{Yekhanin08,DvirGY11,Efremenko12, KoppartyMRS17}, and, to the best of our knowledge,  all known constructions of Hamming LDCs are achieved by linear codes.  As further motivation for studying linear LDCs, lower bounds for linear $2$-query (Hamming) LDCs are useful in polynomial identity testing \cite{DvirS07}, and they are known to imply lower bounds on matrix rigidity \cite{Dvir10}. In addition, a recent work \cite{ChengGHL21} has initiated a systematic study on linear insdel codes.

We first show that $2$-query linear insdel LDCs do not exist.

\begin{restatable}{theorem}{thmLinearLDC}
\label{thm:linearLDC}
For any  $(2, \delta, \eps)$ linear or affine insdel LDC $C: \{0, 1\}^n \to \{0, 1\}^m$, we have $n=O_{\delta, \eps}(1)$.
\end{restatable}

More generally, we show an exponential lower bound for general $2$-query insdel LDCs.

\begin{restatable}{theorem}{thmTwoQueryGeneralLowerBound}\label{thm:2QueryGeneralLowerBound}
    For any $(2,\delta, \eps)$ insdel LDC $C: \{0,1\}^n\rightarrow \{0,1\}^m$, we have $m=\exp(\Omega_{\delta, \eps}(n)).$
    
\end{restatable}

We remark that, as previously mentioned, the lower bound for $2$-query Hamming LDCs from \cite{KerenidisW04} also holds for $2$-query Insdel LDCs. However, that proof uses sophisticated quantum arguments, and an important quest in the area has been providing non-quantum proofs for the same result. Indeed, the proof from \cite{KerenidisW04} was adapted to classical arguments by \cite{Ben-AroyaRW08}, but the arguments still retained a strong quantum-style flavor. Our arguments here do not resemble those proofs and are purely classical.
Furthermore, in contrast to the lower bounds from \cite{KerenidisW04,Ben-AroyaRW08}, our lower bounds in Theorems \ref{thm:linearLDC} and \ref{thm:2QueryGeneralLowerBound} extend to the private-key setting where the encoder and decoder share private randomness. We note that one {\em can} easily obtain private-key Hamming LDCs with $m=\tilde{O}(n)$ and locality $1$ by modifying the construction of \cite{OPS07} --- see details in Appendix \ref{subsec:privkeyinsdel}.

By contrast, for any constants $\epsilon, \delta >0$ our results rule out the possibility of $2$-query Insdel LDCs with $m=\exp(o(n))$ even in the private-key setting.

\subsubsection{Lower bounds for \texorpdfstring{$q\geq 3$}{q>=3} query Insdel LDCs}

We prove the following general bound for $q\geq 3$ queries. 

\begin{restatable}{theorem}{obliviousLB} \label{thm:obliviousLBcombined}
	For any non-adaptive $\tp{q,\delta,\eps}$ insdel LDC $C \colon \set{0,1}^n \rightarrow \set{0,1}^m$ with $q \ge 3$, we have the following bounds.
	\begin{itemize}
	\item For arbitrary adversarial channels,
	\begin{align*}
		m=\begin{cases}
			\exp\tp{\Omega_{\delta, \eps}\tp{\sqrt{n}}} \text{ for } q=3 \text{; and } \\
			\exp\tp{\Omega\tp{\frac{\delta}{\ln^2(q/\eps)} \cdot \tp{\eps^3 n}^{1/(2q-4)}}} \text{ for } q \ge 4.
		\end{cases}
	\end{align*}

	\item For the private-key setting, 
	\begin{align*}
		m = \exp\tp{\Omega\tp{\frac{\delta}{\ln^2(q/\eps)} \cdot \tp{\eps^3 n}^{1/(2q-3)}}}.
	\end{align*}
	\end{itemize}

\end{restatable}

As a comparison, for general Hamming LDCs the best known lower bounds for $q \ge 3$ in~\cite{Woodruff07} give $m=\Omega(n^2/\log n)$ for $q=3$, and $m=\Omega({n}^{1+1/\ceil{(q-1)/2}}/\log n)$ for $q>3$. 
Thus, in the constant-query regime, the bounds from Theorem \ref{thm:obliviousLBcombined} are essentially exponential in the existing bounds for Hamming LDCs. Moreover, these bounds also give a separation between constant-query Hamming LDCs, which can have length $\exp(n^{o(1)})$, and constant-query insdel LDCs. 

\paragraph{Lower bounds for adaptive decoders} It is well-known \cite{KatzT00} that a $(q, \delta, \eps)$ adaptive Hamming LDC can be converted into a non-adaptive $(\frac{|\Sigma|^q -1}{|\Sigma|-1}, \delta,\eps)$ Hamming LDC, and also into a non-adaptive 
$(q, \delta, \eps/|\Sigma|^{q-1})$ Hamming LDC, and hence lower bounds for non-adaptive decoders imply lower bounds for adaptive decoders, with the respective loss in parameters.  
It is easy to verify that the same reduction works for Insdel LDCs.\footnote{For example, our non-adaptive decoder can pick $r_1,\ldots, r_{q-1} \in \Sigma$ randomly and simulate the adaptive $(q,\delta, \epsilon)$-decoder responding to the first $q-1$ queries with $r_1,\ldots, r_{q-1}$. This allows the non-adaptive decoder to extract a set $(j_1,\ldots, j_q)$ of queries representing the set of queries that the adaptive decoder would have asked given the first $q-1$ responses. The queries $(j_1,\ldots, j_q)$ can then  be asked non-adaptively to obtain $y[j_1],\ldots, y[j_q]$. With probability $|\Sigma|^{-q+1}$ we will have $y[j_i] = r_i$ for each $i \leq q-1$ and we can finish simulating the adaptive decoder to obtain a prediction $x_i$ which will be correct with probability at least $\frac{1}{2} + \eps$. Otherwise, our non-adaptive decoder randomly guesses the output bit $x_i$. Thus, the non-adaptive decoder is successful with probability at least $\frac{1}{2} + \epsilon |\Sigma|^{-q+1}$. } In particular our lower bounds imply the respective lower bounds for adaptive decoders.

\begin{corollary} \label{cor:adaptive}
For any (possibly adaptive) $\tp{q,\delta,\eps}$ insdel LDC $C \colon \set{0,1}^n \rightarrow \set{0,1}^m$ with $q \ge 3$, we have the following bounds.
	\begin{itemize}
	\item For arbitrary adversarial channels,
	\[m=
	\begin{cases}
	 \exp(\Omega_{\delta, \eps}(\sqrt{n})) \text{ for } q=3 \text{; and } \\
	\exp \tp{ \Omega\tp{\frac{\delta}{(q+\ln(q/\eps))^2} \cdot \tp{\eps^3 n}^{1/(2q-4)}}} \text{ for } q \ge 4.
		\end{cases}
	\]

	    \item For the private-key setting, 
	    \begin{align*}
	m = \exp\tp{\Omega\tp{\frac{\delta}{(q+\ln(q/\eps))^2} \cdot \tp{\eps^3 n}^{1/(2q-3)}}}.
	\end{align*}
	\end{itemize}
\end{corollary}
Corollary \ref{cor:adaptive} is obtained by plugging $\epsilon' = \epsilon/2^{q-1}$ into Theorem \ref{thm:obliviousLBcombined} and applying the average case reduction from a $(q,\delta,\epsilon)$ (adaptive) Insdel LDC to a $(q,\delta, \epsilon/2^{q-1})$ (non-adaptive) Insdel LDC \cite{KatzT00}. Corollary~\ref{cor:adaptive} also implies lower bounds in regimes where $q$ is slightly super-constant (but $o(\log n)$).

\begin{corollary} \label{cor:super-constant-query-obl}
For any (possibly adaptive) $(q,\delta,\eps)$ insdel LDC $C \colon \set{0,1}^n \rightarrow \set{0,1}^m$, the following bounds hold. 
\begin{itemize}
\item If $q = O\tp{\log\log n}$, then $m = \exp\tp{\exp( \Omega_{\delta,\eps}(\log n/\log\log n))}$.
\item If $q=\log n/(2c\log\log n)$ for some $c > 3$, then
$m = \exp(\Omega(\log^{c-2}n))$. In turn, if $m=\poly(n)$, then $q=\Omega(\log n/\log\log n)$.
\end{itemize}
	Moreover, the lower bounds hold even in private-key settings.
\end{corollary}

We remark that the lower bound for $q=O(\log\log n)$ queries is even larger than the Hamming LDC upper bound of $\exp(\exp((\log n)^{1/t}  (\log\log n)^{1-1/t}))$ due to~\cite{Yekhanin08,DvirGY11,Efremenko12} for $q=2^t$ being a constant number of queries.

Furthermore, we get a super-polynomial lower bound even if $q = \log n/(8\log\log n)$. Thus to get any polynomial length Insdel LDC one needs $q=\Omega(\log n/\log \log n)$. This can be compared to the Insdel LDC constructions in ~\cite{Ostrovsky-InsdelLDC-Compiler, BlockBGKZ20}, which give $m=o(n^2)$ with $q=(\log n)^C$ for some $C>2$ (or to the private-key Insdel LDC construction in \cite{ChengLZ20, block2021private} which gives constant rate $m=\Theta(n)$ and $q=(\log n)^C$ for some $C>2$). Both the lower bound and the upper bound are for an adaptive Insdel LDC, so our lower bound on the query complexity almost matches the upper bound for polynomial length Insdel LDCs. This also implies that there is a ``phase transition'' phenomenon in the $q=\polylog(n)$ regime, where the length of the Insdel LDC transits from super-polynomial to polynomial.

\subsubsection{Implications to lower bounds for Insdel LCCs}\label{sec:intro-LCCs}

As mentioned above, the lower bounds for Insdel LDCs extend to Insdel LCCs due an analogue to Insdel errors of a reduction \cite{KaufmanV10,bhattacharyya2016lower} between the two notions in the Hamming error model. More specifically, in  \cite{KaufmanV10,bhattacharyya2016lower}, the authors show that in the standard Hamming error case, any $q$-query LCC  can be converted into a $q$-query LDC with only a constant loss in rate, and preserving the other relevant  parameters. 
In \cite{KaufmanV10}, Kaufman and Viderman show that the two notions are not equivalent in some specific sense, as there exist LDCs that are not LCCs \cite{KaufmanV10}.

We start with a formal definition.

\begin{restatable}{definition}{DefInsdelLCC}[Insdel Locally Correctable Codes (Insdel LCCs)] \label{def:InsdelLCC}
Fix an integer $q$ and constants $\delta \in [0,1]$, $\eps \in (0, \frac{1}{2}]$. We say $C \colon \set{0,1}^n \rightarrow \Sigma^m$ is a \emph{$(q,\delta,\eps)$-locally correctable insdel code} if there exists a probabilistic algorithm $\Dec$ such that:
\begin{itemize}
\item For every $x \in \set{0,1}^n$ and $y \in \Sigma^{m'}$ such that $\ED\tp{C(x), y} \le \delta \cdot 2m$, and for every $i \in [m]$, we have
\begin{align*}
\Pr\left[ \Dec(y, m', i) = C(x)_i \right] \ge \frac{1}{2} + \eps,
\end{align*}
where the probability is taken over the randomness of $\Dec$, and $\ED\tp{C(x),y}$ denotes the minimum number of insertions/deletions necessary to transform $C(x)$ into $y$.
\item In every invocation, $\Dec$ reads at most $q$ symbols of $y$. We say that $\Dec$ is \emph{non-adaptive} if the distribution of queries of $\Dec(y, m', i)$ is independent of $y$.
\end{itemize}
\end{restatable}

We note that if $C$ is a linear/affine insdel LCCs then $C$ is also an insdel LDC. Indeed, linear/affine codes are {\em systematic} codes, and hence the message bits appear as codeword bits. This is also the case in the private-key setting. For completeness, we include a proof in Appendix \ref{sec:LCCtoLDC}.
Hence our lower bounds about linear/affine insdel LDCs apply to linear/affine insdel LCCs, even in the private-key setting.

Our results can be extended to non-linear LCCs and LDCs (but not in the private-key setting), using  the following theorem, which we prove in Appendix \ref{sec:LCCtoLDC} via small adaptations to the proof of \cite{bhattacharyya2016lower}.

\begin{theorem}\label{thm:LCCtoLDCreduction}
Let $C: \{0,1\}^k \rightarrow \Sigma^m$ be a $(q,\delta, \eps)$-insdel LCC, then there exists a $(q, \delta, \eps)$-insdel LDC $C':\{0,1\}^{k'} \rightarrow \Sigma^m$ with
\[k' = \Omega\tp{\frac{k}{\log(1/\delta)}}.\]
\end{theorem}

We conclude the following about Insdel LCCs.

\begin{corollary}

The asymptotic lower bounds for Insdel LDCs in Theorems~\ref{thm:linearLDC}, \ref{thm:2QueryGeneralLowerBound}, \ref{thm:obliviousLBcombined} (for arbitrary adversarial channels only), and the respective corollaries, also hold for Insdel LCCs. 
\end{corollary}

\subsubsection{A stronger version of the lower bounds}

Our lower bounds above hold against adversarial channels, where the channel may first inspect the codeword and then introduce worst-case error patterns. Our techniques, however, work for a more innocuous channel, namely one that is oblivious to both the codeword and the decoder. We formalize the stronger version of our results below. We believe in this form they may be more easily applicable to other settings.


\begin{definition}[channel]
A channel $\mathfrak{C}$ for $m$-bit strings is a Markov chain on $\Omega=\set{0,1}^m$. Equivalently, it is a collection of distributions $\set{\mathfrak{C}(s) \colon s \in \set{0,1}^m}$ over $\set{0,1}^m$, where the output of $\mathfrak{C}$ on input $s$ is a random string $s'$ distributed according to $\mathfrak{D}(s)$.
\end{definition}

We remark that this definition does allow the output of a channel to depend on its input. However, since the channel is fixed for any decoding algorithm, the following notion of ``decodable on average'' is well-defined. We note that a similar notion would not make sense for an adversarial channel, as the channel can be adaptive to the decoding strategy.

\begin{definition}[decoding on average]
A code $C\colon \set{0,1}^n \rightarrow \set{0,1}^m$ is \textup{$(q,\delta,\eps)$-locally decodable on average for channel $\mathfrak{C}$} if there exists a probabilistic algorithm $\Dec$ such that
\begin{itemize}
	\item For every $i \in [n]$, we have
	\begin{align*}
	\Pr_{\substack{x \in \set{0,1}^n \\ y \sim \mathfrak{C}(C(x))}}\left[ \Dec(y, m, i) = x_i \right] \ge \frac{1}{2} + \eps,
	\end{align*} 
	where the probability is taken over the uniform random choice of $x \in \set{0,1}^n$, the randomness of $\mathfrak{C}$, and the randomness of $\Dec$. 
	\item $\Dec$ makes at most $q$ queries to $y$ in each invocation. 
\end{itemize}
\end{definition}

We now state our lower bound in the strongest form. Its proof can be obtained by inspecting the proof of Theorem~\ref{thm:obliviousLBcombined} (for the private-key setting).
\begin{theorem} \label{thm:lb-channel}
	There exists a channel $\mathfrak{C}$ for $m$-bit strings such that: 
	\begin{itemize}
	\item For every $s \in \set{0,1}^m$, $\Pr_{s'\sim \mathfrak{C}(s)}[\ED\tp{s', s} > \delta\cdot 2m] < \negl(m)$.
	\item Let $C \colon \set{0,1}^n \rightarrow \set{0,1}^m$ be a code that is $(q,\delta,\eps)$-locally decodable on average for $\mathfrak{C}$. Then for $q\ge 3$ we have $m = \exp\tp{\Omega_{q,\delta,\eps}(n^{1/(2q-3)})}$.
	\end{itemize}
\end{theorem}

\subsection{Overview of techniques}
Here we give an informal overview of the key ideas and techniques used in our proofs. We always assume a non-adaptive decoder in the following discussion.

\paragraph{Prior strategies for Hamming LDC lower bounds} 
We start by discussing the proof strategies in lower bounds for Hamming LDCs. Essentially all such proofs\footnote{Except the proof in \cite{GalM12} which gives a lower bound for $3$-query Hamming LDC in a special range of parameters.} begin by observing that the code needs to be \emph{smooth} in the sense that for any target message bit, the decoder cannot query a specific index with very high probability. Using this property, one can show that if we represent the queries used by the decoder as edges in a hypergraph with $m$ vertices, then for any target message bit the hypergraph contains a matching of size $\Omega(m/q)$. The key idea in the proof is now to analyze this matching, where one uses various tools such as (quantum) information theory \cite{KatzT00, KerenidisW04, Woodruff07}, matrix hypercontractivity \cite{Ben-AroyaRW08}, combinatorial arguments \cite{KatzT00, BhattacharyyaCG20}, and reductions from $q$-query to $2$-query \cite{Woodruff07, Woodruff12}.

For our proofs, however, the matching turns out to be not the right object to look at. Indeed, by simply analyzing the matching it is hard to prove any strong lower bounds for $q \geq 3$, as evidenced by the lack of progress for Hamming LDCs. Intuitively, a matching does not capture the essence of insdel errors (e.g., position shifts), which are strictly more general and powerful than Hamming errors. Therefore, we instead need to look at a different object.

\paragraph{The Good queries} For a $q$-query insdel LDC, the correct object turns out to be the set of all \emph{good} $q$-tuples in the codeword that are potentially useful for decoding a target message bit. When we view the bits in the codeword as functions of the message, we define a $q$-tuple to be good for the $i$'th message bit if there exists a Boolean function $f: \{0, 1\}^q \to \{0, 1\}$ which can predict the $i$'th message bit with a non-trivial advantage (e.g., with probability at least $1/2+\eps/4$, see Definition~\ref{def:good-set}), using these $q$ bits. It is a straightforward application of information theory (e.g., Theorem 2 in~\cite{KatzT00}) that any $q$-tuple cannot be good for too many message bits. Therefore, intuitively, if we can show that any message bit requires a lot of good tuples to decode, then we can conclude that there must be many tuples and thus the codeword must be long. In the extreme case, if we can show that any message bit requires a \emph{constant fraction} of all tuples to decode, then we can conclude that there can be at most a constant number of message bits, regardless of the length of the codeword.  

Towards this end, we consider the effect of insdels on the tuples. Suppose the decoder originally queries some $q$-tuple $A$. After some insdels (e.g., deletions) the positions of the tuples will change, and the actual tuple the decoder queries using $A$ now may correspond to some other tuple $B$ in the original codeword. $B$ may not be a good tuple, in which case it's not useful for decoding the message bit. However, since the decoder always succeeds with probability $1/2+\eps$ when the number of errors is bounded, the decoder should still hit good tuples with a decent probability (e.g., $3\eps/2$). Intuitively, this already implies in some sense that there should be many good tuples, except that this depends on the decoder's probability distribution. For example, if the decoder queries one tuple with probability $1$, then for any fixed error pattern one just needs to make sure that one specific tuple is good. 

To leverage the above point, we turn to a probabilistic analysis and use random errors. Specifically, we carefully design a probability distribution on the insdel errors. For any $q$-tuple $A$, this distribution also induces another probability distribution for the $q$-tuple $B$ which $A$ corresponds to in the original codeword. The key ingredient in all our proofs is to design the error distribution such that it ensures certain nice properties of the induced distribution of any $q$-tuple, which will allow us to establish our bounds. This can be viewed as a conceptual contribution of our work, as we have reduced the problem of proving lower bounds of insdel LDCs to the problem of designing appropriate error distributions.

\paragraph{Designing the insdel error distribution} What is the best insdel error distribution for our proof? It turns out the ideal case for the induced distribution of a $q$-tuple is the uniform distribution. Indeed, the hitting property discussed above implies that for any message bit, there is at least one $q$-tuple in the support of the decoder's queries which would still be good with constant probability under the induced distribution. If we can design an error distribution such that for \emph{any} $q$-tuple, the induced distribution is the uniform distribution on all $q$-tuples, this means that for any message bit, there are at least a constant fraction of all $q$-tuples that are good for this bit, which would in turn imply that there can be at most a constant number of message bits.

However, it appears hard to design an error distribution with the above property even for $q=2$, since we have a bound on the total number of errors allowed, and errors allocated to one tuple will affect the number of errors available for other tuples. Instead, our goal is to design the error distribution such that the induced  distribution of any $q$-tuple is as ``close'' to the uniform distribution as possible. We first illustrate our ideas for the case of $q=2$.

\paragraph{The case of $q=2$} A simple idea is to start with a random number (up to $\Omega(m)$) of deletions at the beginning of the codeword, we call this deletion {\em type $\mathbf{1}$}. This results in a random shift of any pair of indices. However, a crucial observation is that the \emph{distance} between any pair of indices stays the same (for a pair of indices $i, j \in [m]$, their distance is $|i-j|$), which makes the induced distribution far from being uniform. Indeed, under such error patterns the Hadamard code seems to be a good candidate for insdel LDC. This is because any codeword bit of the Hadamard code is the inner product of a vector $v \in \{0, 1\}^n$ with the message, and to decode the $i$'th message bit the decoder queries a pair of inner products for $v$ and $v+e_i$ ($e_i$ is the $i$'th standard basis vector) where $v$ is a uniform vector. If we arrange the codeword bits in the natural lexicographical order according to $v$, then all pairs used in queries for the $i$'th message bit have a fixed distance of $2^{i-1}$. In fact we show in the appendix that a simple variant of the Hadamard code does give a LDC under deletion type $\mathbf{1}$. However, our Theorem~\ref{thm:linearLDC} implies that it is not an insdel LDC in general. The point here is that we need a different operation to change the distance of any pair, which is a phenomenon unique to insdel LDC and never happens in Hamming LDC. 

To achieve this, we introduce \emph{random deletions} of each message bit on top of the previous operation. Specifically, imagine that we fix a constant $p < \delta$ and delete each bit of the codeword independently with probability $p$. Under this error distribution, any pair of queries with distance $d$ will correspond to a pair with distance $\frac{d}{1-p}$ in expectation (since we expect to delete $p$ fraction of bits in any interval). However, the independent deletions lead to a concentration around the mean. Thus the probability of any distance around $\frac{d}{1-p}$ is $\Theta(\frac{1}{\sqrt{d}})$ and the distribution resembles that of a binomial distribution (it is called a negative binomial distribution), which is not flat enough compared to the uniform distribution. Therefore, we add another twist by first picking the parameter $p$ uniformly from an interval (e.g., $[\frac{\delta}{8}, \frac{\delta}{4}]$) and then delete each bit of the codeword independently with probability $p$. We call this deletion {\em type $\mathbf{2}$}. Somewhat magically, the compound distribution now effectively ``flattens'' the original distribution, and we can show that the probability mass of any distance is now $O(\frac{1}{d})$. Intuitively, this is because the distance in the induced distribution is now roughly equally likely to appear in the interval $[\frac{d}{1-\delta/8}, \frac{d}{1-\delta/4}]$. Combined with the deletions at the beginning, we can conclude the following two properties for any pair with original distance $d$ in the induced distribution: (1) The probability mass of any element in the support is  $O(\frac{1}{md})$, and (2) With high probability, the corresponding pair will have distance in $[d, cd]$ for some constant $c=c(\delta, \eps)$ (See Lemma~\ref{lem:distribution} for the formal statement).

While this is not exactly the uniform distribution, it is already enough to establish non-trivial bounds. To do this, we divide all pairs of queries into $O(\log m)$ intervals based on their distances, where the $j$'th interval $P_j$ consists of all pairs with distance in $[c^{j-1}, c^j)$. By the hitting property discussed before, for any message bit, there is at least one $q$-tuple in the support of the decoder's queries which is still good with constant probability under the induced distribution. By (1) and (2) above, there must be at least $\Omega(md)$ good pairs with distance in $[d, cd]$, and this further implies that there exists a $j$ such that $P_j$ contains a constant fraction of good pairs. Now a packing argument implies that $n=O(\log m)$.

We remark that the random deletion channel (described above) that we use to establish the lower bound does not depend on anything about the codeword or the entire coding and decoding scheme. Thus, in contrast to the same bound for Hamming LDC, our lower bound continues to apply in private-key settings where the encoder and decoder share secret randomness. 

\paragraph{Linear $2$-query LDC} The case of linear/affine codes is more involved. Here, we first use Fourier analysis to argue that if a pair of codeword bits is good for decoding a message bit, then the message bit must have non-trivial correlation with some parity of the codeword bits. However, since the code itself is linear or affine, this non-trivial correlation must be $1$. By the hitting property, for any $i$'th message bit there exists a $j_i$ such that a constant fraction of the pairs in $P_{j_i}$ are good for $i$. By rearranging the message bits, without loss of generality we can assume that $j_1 \leq j_2 \leq \cdots \leq j_n$.

Now, for any $i$ and $P_{j_i}$ we have two cases: the message bit can have correlation $1$ either with a single codeword bit, or with the parity of the two codeword bits. By averaging, at least one case consists of a constant fraction of the pairs in $P_{j_i}$. By another averaging, at least a constant fraction of the message bits fall into one of the above cases, so eventually we have two cases: (a) a constant fraction of the message bits each has correlation 1 with a constant fraction of all codeword bits, or (b) a constant fraction of the message bits each has correlation 1 with the parity of a constant fraction of the pairs in $P_{j_i}$. 

The first case is easy since any codeword bit cannot simultaneously have correlation $1$ with two different message bits, hence this implies we can only have a constant number of message bits. The second case is harder, where we use a delicate combinatorial argument to reduce to the first case. Specifically, for any such message bit $i$ we can consider the bipartite graph $G_i$ on $2m$ vertices induced by the good pairs in $P_{j_i}$, thus any such graph has bounded degree (since the distance of the pairs is bounded) and is dense in the sense that the edges take up a constant fraction of all possible edges. For simplicity let us assume that having correlation $1$ means that the two bits are the same as functions. Roughly, we use the dense property of these graphs to show the following: (c) there is an index $i=\Omega(n)$ and a right vertex $W \in G_i$ which is connected to a set $T$ of $\Omega(c^{j_i})$ left vertices in $G_i$, and (d) there are $\Omega(n)$ indices $i' \leq i$ such that in each $G_{i'}$, the same set $T$ is connected to a set $U_{i'}$ of $\Omega(c^{j_i})$ neighbors. By (c), all the codeword bits in $T$ must be the same, and they are all contained in an interval of length $c^{j_i}$. Then by (d), all the codeword bits in $U_{i'}$ for different $i'$ must be disjoint, since the parity of them with some bits in $T$ equals a different message bit. Now notice that for any $i' \leq i$, all pairs in $P_{j_i'}$ have distance at most $c^{j_i'} \leq c^{j_i}$. This implies all the bits of all $U_{i'}$ are contained in an interval of length $2c^{j_i}$, which readily gives that $n=O(1)$.

\paragraph{The case of $q \geq 3$} We now generalize the above strategy to the case of $q \geq 3$. Consider the case of $q=3$ for example. Now any query is a triple and we use $(d_1, d_2)$ to stand for the distances of the two adjacent intervals in the query. If we can show similar properties as before, i.e., for any triple with distance $(d_1, d_2)$ in the induced distribution: (1) The probability mass of any element is  $O(\frac{1}{md_1d_2})$, and (2) With high probability, the corresponding triple will have distance $(d_1', d_2')$ such that $d_1' \in [d_1, cd_1], d_2' \in [d_2, cd_2]$ for some constant $c=c(\delta, \eps)$, then a similar argument would yield the bound of $n=O(\log^2 m)$, and for general $q$ (at least constant $q$) the bound of $n=O(\log^{q-1} m)$.

However, unlike the case of $q=2$, another tricky issue arises here. The issue is that with the error distribution discussed above, while we can ensure that for any \emph{pair} of indices in the $q$-tuple, its marginal distribution behaves as before, the joint distribution of the $q$-tuple in the induced distribution behaves differently than what we expect. The reason is that (e.g., for $q=3$) the two intervals with distance $d_1$ and $d_2$ are correlated under the error distribution. Specifically, the random deletion of each codeword bit again leads to a concentration phenomenon, thus conditioned on the number of deletions in the first interval, the parameter $p$ is no longer uniformly distributed in the interval $[\frac{\delta}{8}, \frac{\delta}{4}]$, but rather pretty concentrated in a much smaller interval. This in turn affects the induced distribution of the second interval. Specifically, under this error distribution the bound on the probability in (1) becomes $O(\frac{\sqrt{d}}{md_1d_2})$, where $d=d_1+d_2$. If we simply apply this bound, it will lead to (coincidentally or uncoincidentally) almost exactly the same bound as for Hamming LDC, thus we don't get any significant improvement.

To get around this and prove strong lower bounds for insdel LDCs, we introduce additional random deletion processes to ``break'' the correlations discussed above. Towards this, we add another layer of deletions on top of the previous two operations: we first divide the codeword evenly into blocks of size $s$, and then for each block, we independently pick a parameter $p$ uniformly from $[\frac{\delta}{8}, \frac{\delta}{4}]$ and delete each bit of this block independently with probability $p$. The idea is that, if for a $3$-query it happens that one block is completely contained in one interval, then since the deletion process in that block is independent of the other blocks, the induced distribution of that interval is also more or less independent of the other interval. 

However, this comes with another tricky issue: how to pick the size $s$. If $s$ is too large, then for queries with small intervals, both intervals can be contained in the same block, and the deletion process would be exactly the same as before, which defeats the purpose of using blocks. On the other hand, if $s$ is too small, then for queries with large intervals, the concentration and correlation phenomenon will happen again, which also defeats the purpose of using blocks. Since the intervals of the queries can have arbitrary distance, our solution is to actually use $O(\log m)$ layers of deletions, where for the $j$'th layer we use a block size of say $2^j$. This ensures that for any query there is an appropriate block size, and in the analysis we can first condition on the fixing of all other layers, and argue about this layer.

Yet there is another price to pay here: since we are only allowed at most $\delta m$ deletions, in each layer we cannot delete each bit with constant probability. Therefore for these layers we need to pick a parameter $p$ uniformly from $[\frac{\delta}{8 \log m}, \frac{\delta}{4 \log m}]$. We call this deletion {\em type $\mathbf{3}$.} This blows up our bound of the probability in (1) by a $\polylog$ factor (see Corollary~\ref{cor:distribution-obl} for a formal statement), and we get a bound of $n=O(\log^{2q-3} m)$.

We note that in all the discussions so far, our error distributions do not depend on anything about the codeword or the entire coding and decoding scheme, thus all these results apply in settings where the encoder and decoder share secret randomness (private-key), which makes our lower bounds stronger. On the other hand, by exploiting the decoder's strategy, we can actually improve our bounds for the case of $q \geq 3$ (but the improved lower bounds no longer hold in the private-key setting). This time, we add another $O(\log m)$ layers of deletions on top of the previous three operations, where for the $j$'th layer we again use a block size of say $2^j$. However, for these $O(\log m)$ layers the deletion parameter $p$ is not picked from $[\frac{\delta}{8\log m}, \frac{\delta}{4\log m}]$, but rather uniformly from $[\frac{\delta p_j}{8}, \frac{\delta p_j}{4}]$, where $p_j$ is the probability that the decoder uses a query whose first interval has distance in $[2^{j-1}, 2^j)$. We call this deletion {\em type $\mathbf{4}$}. Notice that since $\sum_j p_j=1$ the expected number of total deletions for this operation is still at most $\frac{\delta m}{4}$.    

To get some intuition of why this helps us, consider the extreme case where all the queries used by the decoder have exactly the same distance for the first interval. Since there is no other distance for the first interval, we should not assign any probability mass of deletions to blocks of a different size, but should instead use the same block size, and delete each bit with probability $p$ chosen uniformly from say $[\frac{\delta}{8}, \frac{\delta}{4}]$. This corresponds to the case where some $p_j=1$, and the above strategy is a natural generalization. In the meantime, we still need all previous deletion types to take care of the other intervals. We show that under this deletion process we can replace one $\log m$ factor in the probability of (1) by $1/p_j$ (see Corollary~\ref{cor:distribution-adv} for a formal statement), and overall this leads to a bound of $n=O(\log^{2q-4} m)$ for $q \geq 3$.

\subsection{Open questions and subsequent work}
\paragraph{Better lower bound} 

In \cite{BlockBGKZ20}, a subset of the authors raised the conjecture that  in fact constant-query insdel LDCs do not even exist, in stark contrast to the Hamming case, where the classical Hadamard code is a basic example of a $2$-query LDC.  In subsequent work, Gupta \cite{gupta2024constant} confirmed the  conjecture, by first re-interpreting our techniques for the $2$-query case, and then generalizing to any $O(1)$-queries. For a survey on the current insdel LDCs landscape, we also refer the reader to \cite{survey-grigorescu25}.


\paragraph{Relaxed Insdel LDCs/LCCs} Relaxed (Hamming) LDCs/LCCs are variants in which the decoder is allowed some small probability of outputing a ``don't know'' answer, while it should answer with the correctly decoded bit most of the time. \cite{Ben-SassonGHSV06} proposed these variants and gave constructions with constant query complexity and codeword length $m=n^{1+\eps}.$ More recently \cite{GurRR18} extended the notion to LCCs, and proved similar bounds, which are tight \cite{GurL21}. An open problem here is to understand tight bounds for the relaxed insdel variants of LDCs/LCCs.
A follow-up work by Block, Blocki, Cheng, Grigorescu, Li, Zheng and Zhu \cite{block2023relaxed} studied this topic and gave an exponential lower bound for one kind of relaxed insdel LDCs.

\paragraph{Larger alphabet size} We believe our proofs generalize to larger alphabet sizes, and leave the precise bounds in terms of the alphabet size as an open problem. All the above directions may also be asked for larger alphabet sizes.

\subsection{Further discussion about related work}

\cite{OPS07} gave private key constructions of LDCs with constant rate $m=\Theta(n)$ and locality $\polylog(n)$. \cite{BlockiKZ19} extended the construction from \cite{OPS07} to settings where the sender/decoder do not share randomness, but the adversarial channel is resource bounded i.e., there is a safe-function that can be evaluated by the encoder/decoder but not by the channel due to resource constraints (space, computation depth, etc.). By contrast, in the classical setting with no shared randomness and a computationally unbounded channel there are no known constructions with constant rate $m=\Theta(n)$ and locality $\polylog(n)$. \cite{block2021private} applied the \cite{BlockBGKZ20} compiler to the private key Hamming LDC of \cite{OPS07} (resp. resource bounded LDCs of \cite{BlockiKZ19}) to obtain private key Insdel LDCs (resp. resource bounded Insdel LDCs) with constant rate and $\polylog(n)$ locality.

Insdel LDCs have also been recently studied in {\em computationally bounded channels}, introduced in \cite{Lipton94}. Such channels can perform a bounded number of adversarial errors, but do not have unlimited computational power as the general Hamming channels. Instead, such channels operate with bounded resources: for example, they might only behave like probabilistic polynomial time machines, or log space machines, or they may only corrupt codewords while being oblivious to the encoder's random coins, or they might have to deal with settings in which the sender and receiver exchange cryptographic keys. As expected, in many such  limited-resource settings one can construct codes with strictly better parameters than what can be done generally \cite{GopalanLD04,MicaliPSW05, Guruswami_Smith:2016, ShaltielS16}. LDCs in these channels under Hamming error were studied in \cite{OPS07, HemenwayO08, HemenwayOSW11, HemenwayOW15, BlockiGGZ19,BlockiKZ19}. 

\cite{block2021private} applied the \cite{BlockBGKZ20} compiler to the Hamming LDC of \cite{BlockiKZ19} to obtain a constant rate Insdel LDCs with $\polylog(n)$ locality for resource bounded channels.  The work of \cite{ChengLZ20}  proposes the notion of locally decodable codes with randomized encoding, in both the Hamming and edit distance regimes, and in the setting where the channel is oblivious to the encoded message, or the encoder and decoder share randomness. For edit error they obtain codes with $m=O(n)$ or $m= n \log n$ and $\polylog(n)$ query complexity. However, even in settings with shared randomness or where the channel is oblivious or resource bounded, there are no known constructions of Insdel LDCs with constant locality.
 
Locality in the study of insdel codes was also considered in   \cite{HaeuplerS18}, which constructs explicit synchronization strings that can be locally decoded.

Synchronization strings are powerful ingredients that have been used extensively in constructions of insdel codes. In fact, by combining locally decodable synchronization strings with Hamming LDCs, it seems possible to get similar reductions to Insdel LDCs as those in \cite{OPS07, BlockBGKZ20}.

\subsection{Organization}
In Section \ref{sec:prelim} we give some basic notations and lemmas. In section \ref{sec:LB-2query}, we show our lower bound for 2 query insdel LDCs. In Section \ref{sec:error}, we describe  more general error distributions and their induced properties. In Section \ref{sec:LB_obl} we show our lower bound for $q$-query insdel LDCs for the private key setting. In Section \ref{sec:LB_adv}, we show our stronger lower bound  for   $q$-query insdel LDCs. 

\section{Notation and Preliminary Lemmas}
\label{sec:prelim}
Here we present some common notation and lemmas which we use throughout our proofs.

The indices $i, j, k, \ell$ are reserved for iterators;  
$c, \alpha, \beta, \gamma, \eta$ are reserved for constants; 
$a, b, x, y, z$ are reserved for vectors or strings. For a string $y \in \set{0,1}^m$ and a subset $J \subseteq [m]$ of indices, we write $y_J \coloneqq \set{y_j \colon j \in J}$ for the restriction of $y$ to $J$.

We may assume that decoder always queries exactly $q$ indices. If some query uses a set of indices $Q' \subset [m]$ such that $|Q'| = q' < q$, we can replace $Q'$ by $Q = Q' \cup \set{j_1, \dots, j_{q'-q}}$ where choices of $j_1, \dots, j_{q'-q} \in [m] \setminus Q'$ are arbitrary. In the actual decoding, the decoder will just ignore the extra symbols. 
Given a tuple $\set{k_0, \dots, k_{q-1}}$ with $k_0 < \dots < k_{q-1}$,  we also denote it by $\tp{k_0, d_1, d_2, \dots, d_{q-1}}$ where 
$d_i = k_i - k_{i-1}$ for $i=1,2,\dots,q-1$. Note that this induces a bijection $\psi_{m,q}$ between $\+S_{m,q} = \set{(k, d_1, \dots, d_{q-1}) \colon k, d_1, \dots, d_{q-1} \ge 1, k+d_1+\dots+d_{q-1}\le m}$ and $\binom{[m]}{q}$. Sometimes we will abuse the notation and write $Q \subseteq [m]^q$ while we actually mean the image of $Q$ under $\psi_{m,q}$ (e.g. when we write $A \cap B$ where $A \subseteq \binom{[m]}{q}$ and $B \subseteq \+S_{m,q}$), and vice versa.

Given a distribution $\+D$ over some space $\Omega$, denote by $\supp (\+D) = \set{\omega \in \Omega \colon \+D(\omega) > 0}$ the \emph{support} of $\+D$.

All logs are in base 2 unless otherwise specified. We write $\+H(x) = -x\log x - (1-x)\log(1-x)$ for the binary entropy function, and we use the following upper bound (Proposition \ref{prop:entropy-estimate}). The proof can be obtained via expanding $\+H(x)$ into Taylor series around $x = 1/2$.
\begin{proposition} \label{prop:entropy-estimate}
For $x \in (0, 1/2)$, we have $\+H(1/2+x) \le 1 - (2(\ln 2)^2/3)x^2$. 
\end{proposition}

\paragraph{Basic facts of Fourier analysis.}
We start with a Boolean function from $\{0, 1\}^n \rightarrow \{0, 1\}$ and transform it to the $\{1, -1\}^n \rightarrow \{1, -1\}$ space by the transformation $u\mapsto (-1)^u$ for any bit $u$ in the input or output.
 
Let $f, g$ be two Boolean functions  from Fourier space.
We define their correlation to be $\mathsf{Corr}( f, g) = |\E_x f(x)g(x)| = |\Pr_x[f(x) = g(x)] - \Pr_x[f(x) \neq g(x)]| $.  
For a function $f$, its Fourier expansion is $\sum_{S\subseteq [n]} \hat{f}_S \chi_S(x)$, where $\chi_S(x) = \prod_{i\in S} x_i$ and $\hat{f}_S= \langle f, \chi_S \rangle =  \E_x f(x)\chi_S(x)$.
By this definition, for Boolean functions $f$, we always have $|\hat{f}_S| \leq 1$, since $f(u), \chi_S(u) \in \{-1, 1\}$.

\begin{proposition} \label{prop:Fourier}

Let $f: \{-1, 1\}^n \rightarrow \{-1, 1\}$ and $C:\{-1, 1\}^n \rightarrow \{-1, 1\}^m$ be   arbitrary functions. For every $Q \subseteq \binom{[m]}{q}$, if 
\begin{align*}
\sup_{S\subseteq Q} \left|  \E_x f(x) \prod_{j\in S} y_j  \right| < \frac{\varepsilon}{2^q},
\end{align*}
where $y=C(x)$,
then for any function $g: \{-1 , 1\}^q \rightarrow \{ -1, 1\}$, $ \Pr_x\left[ g(y_Q) = f(x) \right]  < \tp{1+\eps}/2$.

\end{proposition}

\begin{proof}
We know that $g(y_Q) = \sum_{S\subseteq [q]}\hat{g}_S \chi_S (y_Q)$. 
So
\begin{align*}
 \left| \E_x g(y_Q) f(x)\right| &=  \left| \E_x \sum_{S\subseteq [q]} \hat{g}_S \chi_S(y_Q) f(x)\right| \\
 &=  \left| \sum_{S\subseteq [q]} \E_x \hat{g}_S \chi_S(y_Q) f(x)\right| \\
 &\leq \sum_{S\subseteq [q]} \left| \E_x  \hat{g}_S \chi_S(y_Q) f(x) \right|\\
 &\leq 2^q \sup_{S\subseteq [q]} \left|\E_x  \hat{g}_S \chi_S(y_Q) f(x) \right| \\
 &< 2^q \eps/2^q = \eps.
\end{align*}

So
 $ \Pr_x\left[ g(y_Q) = f(x) \right]  < \tp{1+\eps}/2$.
\end{proof}

Our analysis is based on designing a specific error pattern and deriving the necessary properties the decoder needs to have in order to perform well against such errors. In a high level, the error pattern is going to be in the following form. Given a codeword $y \in \set{0,1}^m$, we first obtain the \emph{augmented codeword} $y' \in \set{0,1}^{2m}$ by appending $m$ bits to the end of $y$. These bits may be random, and most often they will be independent and uniformly random bits. Then the augmented codeword undergoes a random deletion process, which we describe in details later in Section~\ref{sec:LB-2query} and Section~\ref{sec:error}. For now, think of it as generating a subset $D \subseteq [2m]$ according to some distribution $\+D$, and then deleting all bits from $y'$ with indices in $D$. Finally, the string output by the deletion process is truncated at length $m$ to obtain the final output $\widetilde{y}$. We will argue that with high probability, $\widetilde{y}$ has length exactly $m$ (i.e. there are at most $m$ deletions in total) and is close to the original codeword $y$ (i.e. only a small number of deletions are introduced to the first half of $y'$). 

One would observe that we could equivalently augment the codeword to length $m$ \emph{after} the deletion process, and indeed this gives the same distribution (if the padded bits are i.i.d.). However, it turns out that our argument becomes cleaner if we view the deletions as if they also occur in the augmented part. Specifically, in the following definition we view the augmented bits as part of the codeword, as it is possible that in some situation they actually help the decoder to decode some message bits.

\begin{definition} \label{def:good-set}
For $i \in [n]$, define the set $\Goody_i$ as
\begin{align*}
\Goody_i \coloneqq \set{Q \in \binom{[2m]}{q} \colon \exists \text{a Boolean function } f \colon \{0, 1\}^q \to \{0, 1\} \text{ such that } \Pr[f(C'(x)_Q) = x_i] \geq \frac{1}{2}+\frac{\eps}{4} },
\end{align*}
where the probability is over the uniform distribution of all messages and any possible randomness in the padded bits.
\end{definition}

For $Q \in \binom{[2m]}{q}$, let $H_Q \subseteq [n]$ be a subset collecting all indices $i$ for which $\Goody_i$ contains $Q$. The following is a corollary to Theorem 2 in~\cite{KatzT00}.

\begin{proposition}  \label{prop:info-theory}

$\forall Q \in \binom{[2m]}{q}, \abs{H_Q} \le q/\tp{1 - \mathcal{H}(1/2+\eps/4)}$.
\end{proposition}
\begin{proof}
Let $I\tp{\*x_{H_Q}; C(\*x)_{Q}}$ denote the mutual information between $\*x_{H_Q}$ and $C(\*x)_{Q}$. We have that
\begin{align*}
	I\tp{\*x_{H_Q}; C(\*x)_{Q}} \le \+H\tp{C(\*x)_{Q}} \le q.
\end{align*}
On the other hand,
\begin{align*}
	I\tp{\*x_{H_Q}; C(\*x)_{Q}} &= \+H\tp{\*x_{H_Q}} - \+H\tp{\*x_{H_Q} \mid C(\*x)_Q} \\
	&\ge \+H\tp{\*x_{H_Q}} - \sum_{i \in H_Q}\+H\tp{x_i \mid C(\*x)_Q} \\
	&\ge \tp{1 - \+H(1/2+\eps/4)} \cdot \abs{H_Q}.
\end{align*}
Rearranging gives the result.
\end{proof}

A deletion pattern is a distribution $\+D$ over subsets of $[2m]$. Let $D \subseteq [2m]$ be a set of deletions. We note that $D$ induces a strictly increasing mapping $\phi_D \colon [2m-|D|] \rightarrow [2m]$, where $\phi_D(i) = \min\set{i' \in [2m] \colon \abs{\overline{D} \cap [i']} \ge i}$, or intuitively the index of $i$ before the deletions are introduced. 

Given $Q=\set{k_0,\dots,k_{q-1}} \in \binom{[m]}{q}$, we denote $Q^D = \set{\phi_D(k_0), \dots, \phi_D(k_{q-1})}$. Note that this is always well-defined when $|D|\le m$. Most often we will work with a random $D \sim \+D$ for some deletion pattern $\+D$. In that case $Q^D$ is a random variable, and sometimes we say that $Q^D$ \emph{corresponds to $Q$} under $\+D$. If the event $Q^D \in \Goody_i$ occurs, where $Q$ is a random query of $\Dec(\cdot, m, i)$, we say that ``$\Dec(\cdot,m,i)$ {\em hits} $\Goody_i$''. In this paper this event will be independent of the string given to $\Dec$ since $\Dec$ is non-adaptive, and $\+D$ will be oblivious to the codeword.

\begin{lemma}\label{lem:hitting}
Given a $(q, \delta, \eps)$ insdel LDC, for any deletion pattern $\+D$ such that $|D \cap [m]| \le \delta m$ and $|D| \le m$ for any $D \in \supp(\+D)$, and any $i \in [n]$, the probability that $\Dec(\cdot,m,i)$ hits $\Goody_i$ is at least $3\eps/2$.
\end{lemma}
\begin{proof}
Consider a uniformly random message $x \in \set{0,1}^n$ and $y = C(x) \in \set{0,1}^m$. Let $y' \in \set{0,1}^{2m}$ be an augment of $y$, and denote by $y^D$ the string obtained by deleting from $y'$ all bits with indices in $D$ and truncating at length $m$. Formally, $y^D_j = y'_{\phi_D(j)}$ for $j=1,\dots,m$. Note that this is well defined if $|D| \le m$. 

Denote by $\+E$ the event ``$\Dec\tp{\cdot,m,i}$ hits $\Goody_i$''. Conditioned on $\overline{\+E}$, the decoder successfully outputs $x_i$ with probability at most $1/2+\eps/4$, by definition of $\Goody_i$ (even in the case where the decoder may output a random function). 

When $\abs{D \cap [m]} \le \delta m$, we have that $\ED(y, y^D) \le \delta \cdot 2m$. By definition of a $(q,\delta,\eps)$ insdel LDC, we have that 
\begin{align*}
    \frac{1}{2} + \eps &\le \Pr\left[ \Dec(y^D, m, i) = x_i \right] \\
    &\le \Pr\left[ \Dec(y^D, m, i) = x_i \mid \+E \right] \cdot \Pr\left[ \+E \right] + \Pr\left[ \Dec(y^D, m, i) = x_i \mid \overline{\+E} \right] \cdot \Pr\left[ \overline{\+E} \right] \\
    &\le \Pr\left[ \+E \right] + \tp{\frac{1}{2} + \frac{\eps}{4} } \cdot \tp{1 - \Pr\left[ \+E \right]}.
\end{align*}
All probabilities above are over $x$, $D$, the randomness of the decoder and any possible randomness in the padded bits. Rearranging gives $\Pr\left[ \+E \right] \geq 3\eps/(2-\eps) \geq 3\eps/2$. 
\end{proof}

We will write $\*U[a,b]$ for the uniform distribution over the interval $[a,b]$. For $n \in \mathbb{N}$ and $p \in [0,1]$, we will write $B(n,p)$ for the binomial distribution with $n$ trials and success probability $p$. When $p$ is a random variable with distribution $\+D$, we will denote the resulting compound distribution by $B(n, \+D)$.

We use the following anti-concentration bound for the compound distribution $B(n, \*U[a,b])$.

\begin{lemma} \label{lem:compound-anticoncentration}
Let $n \in \mathbb{N}$, and $0\le s < t \le 1$. Let $X$ be a random variable following a compound distribution $B(n, \*U[s,t])$. Then for any $0\le k \le n$, we have
\begin{align*}
	\Pr\left[ X = k \right] \le \frac{1}{\tp{t-s}\tp{n+1}}.
\end{align*}
\end{lemma}
\begin{proof}
We can explicitly write the probability as
\begin{align*}
	\Pr\left[ X = k \right] = \frac{1}{t-s}\int_{s}^{t}\binom{n}{k} x^k\tp{1-x}^{n-k} \dif x \le \frac{1}{t-s}\int_{0}^{t}\binom{n}{k} x^k\tp{1-x}^{n-k} \dif x.
\end{align*}
Denoting
\begin{align*}
I_k = \binom{n}{k}\int_{0}^{t}x^k(1-x)^{n-k}\dif x,
\end{align*}
we are going to show that $I_k \le 1/(n+1)$. Integration by parts gives
\begin{align*}
I_k &= \frac{1}{k+1}\binom{n}{k}\tp{x^{k+1}(1-x)^{n-k} \bigg\vert_{0}^{t} + (n-k)\int_{0}^{t}x^{k+1}(1-x)^{n-k-1}\dif x} \\
&= \frac{1}{k+1}\binom{n}{k}t^{k+1}(1-t)^{n-k} + \frac{n-k}{k+1}\binom{n}{k}\int_{0}^{t}x^{k+1}(1-x)^{n-k-1}\dif x \\
&= \frac{1}{n+1}\binom{n+1}{k+1}t^{k+1}(1-t)^{n-k} + \binom{n}{k+1}\int_{0}^{t}x^{k+1}(1-x)^{n-k-1}\dif x \\
&= \frac{1}{n+1}\binom{n+1}{k+1}t^{k+1}(1-t)^{n-k} + I_{k+1}.
\end{align*}
Therefore by telescoping and the Binomial Theorem, we have
\begin{align*}
I_k = \sum_{j=k}^{n}\tp{I_k - I_{k+1}} = \frac{1}{n+1}\sum_{j=k+1}^{n+1}\binom{n+1}{j}t^j(1-t)^{n+1-j} \le \frac{1}{n+1}.
\end{align*}
\end{proof}

\section{Bounds for 2-query Insdel LDCs} \label{sec:LB-2query}
In this section, we prove lower bounds for $2$-query insdel LDCs (Theorem~\ref{thm:linearLDC} and Theorem~\ref{thm:2QueryGeneralLowerBound}). 

We start by describing the error pattern. It is defined via the following random deletion process which is applied to the augmented codeword described in the last section i.e., we obtain the augmented codeword by appending $m$ bits to the end of the codeword. Recall that  after applying the random deletions below we can always truncate the final string back down to $m$ bits.

\paragraph{Description of the error distribution}  
\begin{description}
    \item[Step 1] Pick a real number $\beta\in [\frac{\delta}{8}, \frac{\delta}{4}]$ uniformly at random and then delete each bit $j \in [2m]$ independently with probability $\beta$.
    
    \item[Step 2] Pick an integer $e_2 \in \set{0,1,\dots,\floor{\frac{\delta m }{4}}}$ uniformly at random and delete the first $e_2$ bits.
\end{description}

We remark that equivalently, the process can be thought of as maintaining a subset $D \subseteq [2m]$ of deletions and updating $D$ in each step, and nothing is really deleted until the end of the process. We will sometimes take this view in later discussions. However, for readability we omitted the details as to how this set is updated.

The following proposition bounds the number of deletions introduced by the process.
\begin{proposition} \label{prop:total-error-bound-2query}
Let $D \subseteq [2m]$ be a set of deletions generated by the process. Then we have 
\begin{itemize}
\item $\Pr\left[ |D \cap [m]| > \delta m \right] \le 2^{-\Omega(m)}$,
\item $\Pr\left[ |D| > m \right] \le 2^{-\Omega(m)}$.
\end{itemize}
\end{proposition}
\begin{proof}
Let $D_1 \subseteq D$ be the subset of deletions introduced during Step 1. Since Step 2 introduces at most $\delta m/4$ deletions, it suffices to upper bound the probabilities of $\abs{D_1 \cap [m]} > 3\delta m/4$ and $\abs{D_1} > 3m/4$. Moreover, it suffices to prove the upper bounds for any fixed $\beta \in [\delta/8, \delta/4]$ picked in Step 1. 

For the first item, notice that each bit $j \in [m]$ is deleted independently with probability $\beta \le \delta/4$. Thus by Hoeffding's inequality
\begin{align*}
    \Pr\left[ \abs{D_1 \cap [m]} \ge \tp{\frac{\delta}{4} + \frac{\delta}{2}}m \right] \le \exp\tp{-\frac{\delta^2 m}{2}} = 2^{-\Omega(m)}.
\end{align*}
The proof of the second item follows similarly from Hoeffding's inequality
\begin{align*}
    \Pr\left[ \abs{D_1} \ge \frac{3m}{4} \right] \le \Pr\left[ \abs{D_1} \ge \tp{\frac{\delta}{4} + \frac{1}{8}} \cdot 2m \right] \le \exp\tp{-2\tp{\frac{1}{8}}^2 \cdot 2m} = 2^{-\Omega(m)}.
\end{align*}
\end{proof}

In the following lemma, we fix an arbitrary query $\set{k, \ell} \in \binom{[m]}{2}$ of the decoder, with $ k< \ell$, and represent it as $(k, d)$ where $d=\ell-k$. 

Let $(k', d') \in [2m] \times [2m]$ be the random pair that corresponds to $(k,d)$ under the error distribution (see the discussion before Lemma~\ref{lem:hitting}). It should be clear that we always have $k' \geq k$ and $d' \geq d$. We prove some properties of the distribution of $(k', d')$.

\begin{lemma}\label{lem:distribution}
There exist two constants $c=c(\eps)> 1$ and $c'=c'(\delta)>0$ such that the following holds.
\begin{itemize}
\item The distribution of $(k', d')$ is concentrated in the set $[2m] \times [d, c d]$ with probability $1-\eps$. 
\item Any support of $(k', d')$ has probability at most $\frac{c'}{md}$. 
\end{itemize}
\end{lemma}

\begin{proof}
We prove the concentration result first. We will fix an arbitrary $\beta \in [\delta/8, \delta/4]$. By Hoeffding's inequality, we can take $c_0 = \sqrt{\ln(1/\eps)/2}$ such that for any $n \in \mathbb{N}$ and $p \in [0,1]$,
\begin{align*}
	\Pr_{Y \sim B(n, p)}\left[ Y \ge pn + c_0\sqrt{n}\right] \le \eps.
\end{align*}
Take $c = 8c_0^2 = 4\ln(1/\eps)$. Then $c > 1+\tp{c_0/(1-\beta)}^2$ for any $\beta \le \delta/4 < 1/2$. Let $X$ denote the number of deletions occurred in $[d+1, cd]$, which follows a binomial distribution $B((c-1)d, \beta)$. Then by the choice of $c_0$ we have
\begin{align*}
\Pr\left[ d' > cd \right] \le \Pr\left[ X \ge (c-1)d \right] \le \Pr\left[ X \ge \beta(c-1)d + c_0\sqrt{(c-1)d} \right] \le \eps.
\end{align*}
Note that this holds for any choice of $\beta \le \delta/4$, and thus the concentration result follows.

Now we turn to the anti-concentration result. Denote by $k' \mapsto k$ the event that the $k'$-th bit is retained and has index $k$ after the deletion, and denote by $\tp{k', d'} \mapsto \tp{k, d}$ the event $\tp{k' \mapsto k} \land \tp{k'+d' \mapsto k+d}$. 

Write $\Pr_{S_1}[\cdot]$ for the error distribution after Step 1. Let $X$ be the number deletions occurred in $\set{k'+1, \dots, k'+d'-1}$, which follows a compound distribution $B(d'-1, \*U[\delta/8, \delta/4])$. We have
\begin{align*}
\sum_{k''=0}^{k'}\Pr_{S_1}\left[ \tp{k', d'} \mapsto \tp{k'', d} \right] &= \sum_{k''=0}^{k'} \Pr_{S_1}\left[ k' \mapsto k'' \right] \cdot \Pr_{S_1}\left[ k'+d' \mapsto k''+d \ \middle| \ k' \mapsto k'' \right] \\
&= \sum_{k''=0}^{k'} \Pr_{S_1}\left[ k' \mapsto k'' \right] \cdot \Pr_{S_1}\left[ X = d'-d \right] \cdot \Pr_{S_1}\left[ k'+d'\textup{ is retained} \right] \\
&\le \Pr_{S_1}\left[k'\textup{ is retained} \right] \cdot \frac{8}{\delta} \cdot \frac{1}{d'} \\
&\le \frac{8}{\delta} \cdot \frac{1}{d'}.
\end{align*}

Here the first inequality is due to Lemma~\ref{lem:compound-anticoncentration}. Finally, averaging over $e_2$ gives
\begin{align*}
	\Pr\left[ \tp{k',d'} \mapsto \tp{k,d} \right] &= \frac{8}{\delta m}\sum_{e_2 = 0}^{\delta m/8}\Pr_{S_1}\left[ \tp{k', d'} \mapsto \tp{k+e_2, d} \right] \\
	&\le \frac{8}{\delta m} \sum_{k''=0}^{k'} \Pr_{S_1}\left[ \tp{k',d'}\mapsto \tp{k'', d}\right] \\
	&\le \frac{8}{\delta m} \cdot \frac{8}{\delta} \cdot \frac{1}{d} \\
	&= \frac{64}{\delta^2} \cdot \frac{1}{md}.
\end{align*}
Therefore we can take $c' = 64/\delta^2$.
\end{proof}

Before proceeding to prove the main theorems, we provide a dictionary of notations to facilitate the readers.

\paragraph{Notations.}

The sets $S_i, S, T, U_i, V_i$ are subsets of $[m]$, where $S_i, S, T, U_i$ are used to denote some set of the first indices (namely $k$ for a pair $\set{k, \ell}$ with $k<\ell$), and $V_i$ is used to denote some set of the second indices (namely $\ell$ for a pair $\set{k, \ell}$ with $k<\ell$). We have the following relation: $\forall i, U_i \subseteq T \subseteq S$.  

The set $\Goody_i$ is defined in Definition~\ref{def:good-set}. The sets $P_j, Q_i$ are subsets of $[2m] \times [2m]$, i.e., subsets of the pairs of indices that may or may not be used in the query. $j$ is reserved for the index of some $P_j$.

The set $G_{k,i}$ is a subset of $[n]$, i.e., a subset of some indices of the message bits.

We recall the statement of our main theorem for $2$-query linear insdel LDC.

\thmLinearLDC*

To prove this theorem we first establish the following claim, which works for any $(2, \delta, \eps)$ insdel (even non-linear/affine) LDC. Let $c$ be the constant from Lemma~\ref{lem:distribution}. Consider all pairs of the form $(k, d)$ in $[2m] \times [2m]$, and partition them into $t=\ceil{\log_c (2m)} = O(\log m)$ subsets $\{P_j\}$, where for any $j \in [t]$, $P_j=[2m] \times [c^{j-1}, c^j)$. 

\begin{claim}\label{clm:fraction}
There exists a constant $\gamma=\gamma(\delta, \eps) \leq 1$ such that the following holds for any $(2, \delta, \eps)$ insdel LDC. For any $i \in [n]$, there exists a $j \in [t]$ such that $|P_j \cap \Goody_i| \geq \gamma m c^j$.
\end{claim}
\begin{proof}
Fix any $i \in [n]$. Let $D \subseteq [2m]$ be a random set of deletions generated by the random process. By Proposition~\ref{prop:total-error-bound-2query}, with probability $1-2^{-\Omega(m)} \geq 1-\eps/4$ for any large enough $n$ (and thus also $m$), we have that $\abs{D \cap [m]} \le \delta m$ and $\abs{D} \le m$. Conditioned on this event, $\Dec(\cdot,m,i)$ hits $\Goody_i$ with probability at least $3\eps/2$ by Lemma~\ref{lem:hitting}. Therefore, unconditionally the probability that $\Dec(\cdot,m,i)$ hits $\Goody_i$ is at least $3\eps/2-\eps/4 = 5\eps/4$

(if $|D| > m$ we simply assume that $\Dec(\cdot,m,i)$ never hits $\Goody_i$). This implies that for at least one $(k, d)$ in the support of the queries of $\Dec(\cdot,m,i)$, the corresponding pair $(k', d')$ hits $\Goody_i$ with probability at least $5\eps/4$.

Now by the first item of Lemma~\ref{lem:distribution} and a union bound, $\Dec(\cdot, m, i)$ queries a pair in $\Goody_i \cap \tp{[2m] \times [d, cd]}$ with probability at least $5\eps/4-\eps=\eps/4$. By the second item of Lemma~\ref{lem:distribution}, we must have 
\begin{align*}
\abs{\Goody_i \cap \tp{[2m] \times [d, cd]}} \ge \frac{\eps md}{4 c'}.
\end{align*} 
Choose $j'$ such that $c^{j'-1} \leq d < c^{j'}$. Noticing that $[2m] \times [d, c d] \subseteq P_{j'} \cup P_{j'+1}$, for some $j \in \set{j',j'+1}$ we must have $\abs{\Goody_i \cap P_j} \ge \eps md/(8c')$. Since $d \ge c^{j'-1} \ge c^{j-2}$, we can choose $\gamma=\eps/(8 c' c^2)$ and the claim follows.
\end{proof}

By the definition of $\Goody_i$ and Proposition~\ref{prop:Fourier}, if a pair $\set{k, \ell} \in \Goody_i$ ($k < \ell$), then one of the following cases must happen: (1) $x_i$ has correlation at least $\eps/8$ with $y_k$; (2) $x_i$ has correlation at least $\eps/8$ with $y_\ell$; and (3) $x_i$ has correlation at least $\eps/8$ with $y_k \oplus y_\ell$. However, notice that the code is a linear or affine code, thus every bit in $C(x)$ is a linear or affine function of $x$, which has correlation either $1$ or $0$ with any $x_i$. Furthermore the inserted bits are independent, uniform random bits. Therefore in any of these cases, the correlation must be $1$ and the bits involved must not contain any inserted bit.

Thus, for any $i \in [n]$ and the corresponding $j \in [t]$ guaranteed by Claim~\ref{clm:fraction}, by averaging we also have three cases: (1) $P_j$ has at least $\gamma m c^j/4$ pairs such that the first bit has correlation $1$ with $x_i$; (2) $P_j$ has at least $\gamma m c^j/4$ pairs such that the second bit has correlation $1$ with $x_i$; and (3) $P_j$ has at least $\gamma m c^j/2$ pairs such that the parity of the pair of bits has correlation $1$ with $x_i$.

By another averaging, we now have two cases: either (a) at least $n/4$ of the message bits fall into case (1) or (2) above, or (b) at least $n/2$ of the message bits fall into case (3) above. We prove Theorem~\ref{thm:linearLDC} in each case.

\begin{proof}[Proof of Theorem~\ref{thm:linearLDC} in case (a)]
In this case, without loss of generality assume that there is a subset $I \subseteq [n]$ with $|I| \geq n/4$ such that for any $i \in I$, the corresponding $P_j$ has at least $\gamma m c^j/4$ pairs such that the first bit has correlation $1$ with $x_i$. Notice that any bit in $C(x)$ can be the first bit for at most $c^j$ pairs in $P_j$, this means that there must be at least $\gamma m /4$ different bits in $C(x)$ that has correlation $1$ with $x_i$. Let this set be $V_i$ and we have $|V_i| \geq \gamma m /4$.

Since for each $i \in I$ we have such a set $V_i$, and these sets must be disjoint (a bit cannot simultaneously have correlation $1$ with $x_i$ and $x_{i'}$ if $i \neq i'$), we have

\[\frac{n}{4} \cdot \frac{\gamma m}{4} \leq \sum_{i \in I} |V_i| =\left | \bigcup_{i \in I} V_i \right | \leq m.\]

This gives $n \leq 16/\gamma=O_{\delta, \eps}(1)$.
\end{proof}

\begin{proof}[Proof of Theorem~\ref{thm:linearLDC} in case (b)]
This is the harder part of the proof. Here, there is a subset $I \subseteq [n]$ with $|I| \geq n/2$ such that for any $i \in I$, the corresponding $P_j$ has at least $\gamma m c^j/2$ pairs such that the parity of the pair of bits has correlation $1$ with $x_i$. For each $i \in I$, let the set of these pairs be $Q_i$. Thus $|Q_i| \geq \gamma m c^{j_i}/2$, where for any $i \in I$, $j_i$ is the corresponding index of $P_j$ guaranteed by Claim~\ref{clm:fraction}. Let $|I|=n' \geq n/2$. By rearranging the message bits if necessary, without loss of generality we can assume that $I=[n']$ and $j_1 \leq j_2 \leq \cdots \leq j_{n'}$. Let $S_i$ be the set of all first indices of $Q_i$ which are connected to at least $\gamma c^{j_i}/4$ second indices. Formally, $S_i=\{k: \left| \{d : (k,d) \in Q_i \} \right| \geq \gamma c^{j_i}/4\}$. 

Another way to view this is to consider the bipartite graph $G_i = \tp{[m], [m], Q_i}$ (since the pairs in $Q_i$ can only involve bits in $C(x)$). Then $G_i$ has at least $\gamma mc^{j_i}/2$ edges, and the left and right degrees of $G_i$ are both at most $c^{j_i}$. Now $S_i$ is the subset of left vertices with degree at least $\gamma c^{j_i}/4$.

We have the following claim.

\begin{claim}
For any $i \in [n']$, $|S_i| \geq \gamma m/4$.
\end{claim}
\begin{proof}
Since $|Q_i| \geq \gamma m c^{j_i}/2$, and each index in $S_i$ is connected to at most $c^{j_i}$ other indices, the claim follows by a Markov type argument.
\end{proof}

Now, for any index $k \in [m]$ and any index $i \in [n']$, we define the set $G_{k, i}$ to be the set of all indices $i' \leq i$ such that $k \in S_{i'}$. Formally, $G_{k,i}=\{ i' \leq i: k \in S_{i'} \}$. We have the following claim:

\begin{claim}
There exists a constant $\eta=\eta(\delta, \eps)=\gamma/8$, an index $i \in [n']$ and a set $S \subseteq S_i$, such that
\begin{itemize}
    \item $|S| \geq \eta m$.
    \item For any $k \in S$, we have $|G_{k, i}| \geq \eta n'$.
\end{itemize}
\end{claim}

\begin{proof}
First notice that $\sum_{k \in [m]} |G_{k, n'}| = \sum_{i \in [n']} |S_i|$. For a pair $(i, k)$ with $i \in [n']$ and $k \in [m]$, we say it is good if $k \in S_i$ and $|G_{k, i}| \geq \gamma n'/8$. For any fixed $k \in [m]$, there are at least $|G_{k, n'}|-\gamma n'/8$ indices $i \in [n']$ such that $(i, k)$ is good (this number may be negative, but that's still fine for us). To see this, let $i^*$ be the smallest index such that $\abs{G_{k, i^*}} = \gamma n'/8$ and notice that $|G_{k, n'}|-\gamma n'/8=|G_{k,n'}|-|G_{k,i^*}|$ counts the number of $i$ such that $i^* < i \le n'$ and $k \in S_i$, i.e. the number of good pairs.

Therefore the total number of good pairs is at least

\[\sum_{k \in [m]}\tp{ |G_{k, n'}|- \frac{\gamma n'}{8} } = \sum_{i \in [n']} |S_i|- \frac{\gamma m n'}{8} \geq \frac{\gamma m n'}{8},\]
since for any $i \in [n']$, we have $|S_i| \geq \gamma m/4$.

By averaging, this implies that $\exists i \in [n']$, such that there are at least $\gamma m/8$ good pairs for this fixed $i$. Let $S$ be the set of all good indices of $k$ for this $i$, then we must have $|S| \geq \gamma m/8$ and for any $k \in S$, we have $k \in S_i$ and $|G_{k, i}| \geq \gamma n'/8$. Thus the claim holds.
\end{proof}

Now consider the index $i$ and the set $S$ guaranteed by the above claim. Recall $j_i$ is the index $j$ of $P_j$ corresponding to $i$. We have the following claim.

\begin{claim}\label{clm:correlation}
There exists a set $T \subseteq S$ and two indices $k_0, \ell_0 \in [m]$ such that the following holds.
\begin{itemize}
    \item $|T| \geq \frac{\eta \gamma c^{j_i}}{4}$.
    \item $T \subseteq [k_0, k_0+c^{j_i}]$.
    \item $\forall k \in T$, $C(x)_k \oplus C(x)_{\ell_0}$ has correlation $1$ with $x_i$.
\end{itemize}

\begin{proof}
Consider all pairs of indices $\set{k, \ell} \in Q_i$ ($k < \ell$) with $k \in S$, and view it as a bipartite graph $G=(A,B,E)$ with indices $k$ on the left, and indices $\ell$ on the right. Formally, $G=(A,B,E)$ with $A=\{a_1,\ldots, a_m\}$, $B=\{b_1,\ldots, b_m\}$ and edge $E=\{(a_k, b_\ell): \set{k,\ell} \in Q_i, k<\ell\}$.  Since for any $k \in S$, we have $k \in S_i$, we know that any $a_k$ has degree at least $\gamma c^{j_i}/4$. Notice that there are $m$ right vertices in $B$.  Therefore there must exist an $\ell_0 \in [m]$ such that the node $b_{\ell_0}$ is connected to at least $\eta \gamma c^{j_i}/4$ vertices on the left, and we can let the set of all these vertices be $T=\{ k: (a_{k}, b_{\ell_0}) \in E\}$. Since for any pair in $Q_i$, the parity of this pair of bits has correlation $1$ with $x_i$, we have that  $C(x)_k \oplus C(x)_{\ell_0}$ has correlation $1$ with $x_i$ for all $k \in T$.

Since the vertices in $T$ are all connected to $\ell_0$, and the distance $d=\ell-k$ for all pairs in $Q_i$ is in $[c^{j_i-1}, c^{j_i}]$, we must have that all indices $k \in T$ are in the range $[\ell_0-c^{j_i}, \ell_0]$. Taking $k_0=\ell_0-c^{j_i}$ and the claim follows.
\end{proof}
\end{claim}

Now for any $i' \leq i$, let $U_{i'}=S_{i'} \cap T$, and consider the set $V_{i'}$ of all indices $\ell \in [m]$ such that $\exists k \in U_{i'}$ with $\set{k, \ell} \in Q_{i'}$. In other words, $V_{i'}$ is set of neighbours of $U_{i'}$ in the bipartite graph $\tp{[m], [m], Q_{i'}}$. We have the following claim.

\begin{claim}
For any $i' \leq i$, we have
\begin{itemize}
    \item $V_{i'} \subseteq [k_0, k_0+2c^{j_i}]$.
    \item $|V_{i'}| \geq \gamma |U_{i'}|/4$.
\end{itemize}
\end{claim}

\begin{proof}
Since $U_{i'} \subseteq T$, and every pair of query in $Q_{i'}$ has distance at most $c^{j_{i'}} \leq c^{j_i}$, we have $V_{i'} \subseteq [k_0, k_0+2c^{j_i}]$. Furthermore, since every index in $U_{i'}$ is connected to at least $\gamma c^{j_{i'}}/4$ indices in $V_{i'}$, while every index in $V_{i'}$ is connected to at most $c^{j_{i'}}$ indices in $U_{i'}$, we must have $|V_{i'}| \geq \gamma |U_{i'}|/4$. 
\end{proof}

Now, notice that for any $i' \leq i$, and any $\ell \in V_{i'}$, there exists some $k \in U_{i'} \subseteq T$ such that $C(x)_k \oplus C(x)_{\ell}$ has correlation $1$ with $x_{i'}$. By Claim~\ref{clm:correlation}, $C(x)_k \oplus C(x)_{\ell_0}$ has correlation $1$ with $x_i$. Thus $C(x)_{\ell} \oplus C(x)_{\ell_0}$ has correlation $1$ with $x_i \oplus x_{i'}$, and $C(x)_{\ell}$ has correlation $1$ with $x_i \oplus x_{i'} \oplus C(x)_{\ell_0}$. This means that for any two $i_1, i_2 \leq i$ with $i_1 \neq i_2$, we must have $V_{i_1} \cap V_{i_2} = \emptyset$. Therefore, all the $V_{i'}$'s for different $i'$ must be disjoint. Thus we have the following inequality:

\[\frac{\gamma}{4} \tp{ \sum_{i' \leq i}|U_{i'}| } \leq \sum_{i' \leq i} |V_{i'}| \leq 2c^{j_i}. \]

Notice that $\sum_{k \in T} |G_{k, i}| = \sum _{i' \leq i} |S_{i'} \cap T|= \sum_{i' \leq i}|U_{i'}|$ and $\forall k \in T \subseteq S$, we have $|G_{k, i}| \geq \eta n'$. Thus 

\[\sum_{i' \leq i}|U_{i'}| \geq  \eta n' |T| \geq \frac{\eta^2 \gamma c^{j_i}}{4} n'. \]

Combining the two inequalities, we get $n' \leq 32/(\eta^2 \gamma^2)=2048/\gamma^4$.
Since $n' \geq n/2$. This also implies that $n \leq 2n' = 4096/\gamma^4=O_{\delta, \eps}(1)$.
\end{proof}

Next we prove a simple exponential lower bound for general $2$-query insdel LDCs, i.e. Theorem~\ref{thm:2QueryGeneralLowerBound}. This should serve as a warm-up for the general $q\geq 3$ case.

\thmTwoQueryGeneralLowerBound*

\begin{proof}

    Recall that $t = \ceil{\log_c(2m)}$ and $P_j = [2m] \times [c^{j-1}, c^j)$ for $j \in [t]$. For $j \in [t]$ and $i \in [n]$, we define $\beta_{j,i} = \frac{\abs{P_j\cap \Goody_i}}{\abs{P_j}}$. Since $\abs{P_j} = 2m(c^j-c^{j-1})\leq 2mc^j$, by Claim~\ref{clm:fraction} there is a constant $\gamma=\gamma(\delta, \eps)<1$ such that for any $i\in [n]$, there exists a $j \in [t]$ satisfying $\beta_{j,i}\geq \gamma$. By the Pigeonhole Principle, there exists a $j \in [t]$ such that $\beta_{j,i}\geq \gamma$ for at least $n/t$ different $i$'s. Fix this $j$ to be $j_0$. We have 
    \[\sum_{i=1}^{n} \beta_{j_0,i} \geq \frac{\gamma n }{t}.\]
    
    On the other hand, by Proposition~\ref{prop:info-theory} every pair $(k,d)$ can belong to $\Goody_i$ for at most $2/(1-\+H(1/2+\eps/4))$ different $i$'s. Thus we have
    \[\sum_{i=1}^{n}\abs{P_{j_0}\cap \Goody_i} \leq \frac{2}{1-\+H(1/2+\eps/4)} \cdot \abs{P_{j_0}}.\]
    Altogether this yields
    \begin{align*}
		\frac{\gamma n}{t} \le \sum_{i=1}^{n} \beta_{j_0,i} = \sum_{i=1}^{n}\frac{\abs{P_{j_0}\cap \Goody_i}}{\abs{P_{j_0}}} \leq \frac{2}{1-\+H(1/2+\eps/4)}.
    \end{align*}
    We have $n \le O_{\delta,\eps}(t) = O_{\delta,\eps}(\log m)$ and $m = \exp\tp{\Omega_{\delta,\eps}(n)}$.
\end{proof}

\section{A More General Error Distribution} \label{sec:error}
In this section we describe a general framework for designing error distributions, and instantiate it with two sets of parameters. The error distribution defined in this section will be used in the proof of Theorem~\ref{thm:obliviousLBcombined}. As before the error distribution is applied to the augmented codeword which is obtained by concatenating $m$ bits to the end of the original codeword --- the final codeword can be truncated back down to $m$ bits after applying the random deletions below. 

Given parameters $L \in \mathbb{N}$, $\*s = \tp{s_1, \dots, s_L} \in [2m]^L$ and $\*h = \tp{h_1, \dots, h_L} \in [0,1]^L$ such that 
\begin{align*}
	h \coloneqq \sum_{\ell=1}^{L}h_{\ell} \le \frac{1}{4},
\end{align*}
we consider an error distribution $\+D\tp{L, \*s, \*h}$ defined by the following process.

\paragraph{Description of the error distribution $\+D\tp{L, \*s, \*h}$}  
\begin{description}
	\item[Step 1] The first step introduces deletions through $L$ layers. For the $\ell$-th layer, we first divide $[2m]$ into $\lceil 2m/s_{\ell}\rceil$ consecutive blocks each of size $s_{\ell}$, except for the last block which may have smaller size. For the $b$-th block in layer $\ell$, we pick $q_{\ell, b} \in [0, h_{\ell}\delta]$ uniformly at random (independent of other blocks), and mark each bit in the block independently with probability $q_{\ell, b}$. Finally, we delete all bits which are marked at least once. 
	
	\item[Step 2] Pick $\beta\in [0, \frac{1}{4}]$ uniformly at random and delete each bit independently with probability $\beta\delta$. 
	
	\item[Step 3] Pick an integer $e_2 \in \set{0, 1, \dots, \floor{\frac{\delta m }{4}}}$ uniformly at random and delete the first $e_2$ bits.
\end{description}

By a union bound, after Step 1, each symbol is deleted with probability at most $h\delta$. We thus have the following proposition as an easy consequence of Hoeffding's inequality.

\begin{proposition} \label{prop:total-error-bound}
	Let $D \subseteq [2m]$ be a set of deletions generated by $\+D(L, \*s, \*h)$. Then we have
	\begin{align*}
	\Pr\left[ \abs{D \cap [m]} > \delta m \right] \le \exp\tp{-\frac{\delta^2m}{8}}, \textup{ and } \Pr\left[ |D| > m \right] \le \exp\tp{-\tp{1-\delta}^2m}.
	\end{align*}
\end{proposition}
\begin{proof}
Let $D_2 \subseteq D$ be the subset of deletions introduced during Step 1 and Step 2. Since Step 3 introduces at most $\delta m/4$ deletions, it suffices to upper bound the probabilities of $\abs{D_2 \cap [m]} > 3\delta m/4$ and $\abs{D_2} > m$. Moreover, it suffices to prove the upper bounds after conditioned on an arbitrary set of deletion probabilities $q_{\ell, b} \in [0, h_{\ell}\delta]$ for each $\ell \in [L]$ and $b \le \ceil{2m/s_{\ell}}$, and $\beta \in [0, 1/4]$.

Under the conditional distribution, each bit $j \in [2m]$ is deleted with probability at most $\tp{h+\beta}\delta \le \delta/2$, and these deletions are independent of each other. The Hoeffding's inequality shows that
\begin{align*}
	&\Pr\left[ \abs{D_2 \cap [m]} > \tp{\frac{\delta}{2} + \frac{\delta}{4}}m \right] \le \exp\tp{-2\tp{\frac{\delta}{4}}^2 m} = \exp\tp{-\frac{\delta^2 m}{8}}, \\
	&\Pr\left[ \abs{D_2} > m \right] = \Pr\left[ \abs{D_2} > \tp{\frac{\delta}{2} + \frac{1-\delta}{2}} \cdot 2m \right] \le \exp\tp{-\tp{1-\delta}^2m}.
\end{align*}
\end{proof}

We fix an arbitrary query $Q=\tp{k,d_1, \dots, d_{q-1}}$ of the decoder, and let $\tp{k',d_1',\dots,d_{q-1}'} \in [2m]^q$ be the random tuple that corresponds to $Q$ under the error distribution $\+D(L, \*s, \*h)$ (see the discussion before Lemma~\ref{lem:hitting}). It should be clear that we always have $k' \ge k, d_1' \ge d_1, \dots, d_{q-1}' \ge d_{q-1}$.

Given the query $Q$, we can define for each $i \in [q-1]$ a subset $F_i \subseteq [L]$ of layers as
\begin{align*}
F_i = \set{\ell \in [L] \colon h_{\ell} \neq 0\textup{ and }\frac{d_i}{4} \le s_{\ell} \le \frac{d_i}{2}}.
\end{align*}

The following lemma is a generalization of Lemma~\ref{lem:distribution}.
\begin{lemma} \label{lem:distribution-generalized}
	Suppose that $F_i \neq \varnothing$ for each $i=2,3,\dots,q-1$. The following propositions hold.
	\begin{itemize}
		\item Let $c = 4\ln\tp{q/\eps}$. The distribution of $(k', d_1', \dots, d_{q-1}')$ is concentrated in the set $[2m] \times [d_1, cd_1] \times \dots \times [d_{q-1}, cd_{q-1}]$ with probability $1 - \eps$.
		\item For any $\tp{\ell_2, \dots, \ell_{q-1}} \in F_2 \times \dots \times F_{q-1}$, any support of $\tp{k', d_1', \dots, d_{q-1}'}$ has probability at most 
		\begin{align*}
		\frac{\tp{32/\delta}^q}{md_1}\prod_{i=2}^{q-1}\frac{1}{h_{\ell_i}d_i}.
		\end{align*}
	\end{itemize}
\end{lemma}
\begin{proof}
	For convenience, let $k'_0 = k'$ and $k'_i = k'_0+\sum_{j=1}^i d'_i$. Similar to the proof of Lemma~\ref{lem:distribution}, we will write $k'\mapsto k$ for the event ``the $k'$-th bit is not deleted and has index $k$ after the deletion process'', and write $\tp{k',d_1',\dots,d_{q-1}'}\mapsto\tp{k,d_1,\dots,d_{q-1}}$ for the event $\bigwedge_{i=0}^{q-1}\tp{k_i'\mapsto k_i}$.
	
	To prove the first item, we are going to condition on an set of deletion probabilities (i.e. $q_{\ell,b}$ for each block and $\beta$), and $e_2$ in Step 3. For each $i\in [q-1]$, we consider a random variable $X_i$ denoting the number of deletions introduced to $I_i \coloneqq \set{k_{i-1}'+1,\dots,k_i'-1}$. It always holds that $0 \le X_i \le d_i'-1$. Note that $X_i$ does not depend on the deletions introduced in Step 3. Under the error distribution, each of these bits is deleted independently with probability at most $(h+\beta)\delta \le \delta/2$. Thus, following an analysis similar to the proof of Lemma~\ref{lem:distribution}, the choice of $c = 4\ln(q/\eps)$ guarantees 
	\begin{align*}
	\Pr[d'_i > c d_i] \le \frac{\eps}{q-1}. 
	\end{align*}
	Taking a union bound shows that 
	\begin{align*}
	\Pr\left[ (k', d_1', d_2', \dots, d_{q-1}')\in [2m] \times [d_1, cd_1] \times \dots \times [d_{q-1}, cd_{q-1}] \right] \ge 1 - \eps.
	\end{align*}
	Recall that this holds for any set of deletion probabilities, and thus the first item follows.
	
	We now show the second item: for any $\tp{\ell_1, \dots, \ell_{q-1}} \in F_1 \times \dots \times F_{q-1}$, we have 
	\begin{equation*}
	\Pr\left[\tp{k', d_1', \dots, d_{q-1}'}\mapsto\tp{k,d_1,\dots,d_{q-1}}\right] \leq \frac{\tp{32/\delta}^q}{md_1} \cdot \prod_{i=2}^{q-1}\frac{1}{h_{\ell_i}d_i}.
	\end{equation*}
	
	Denote by $\Pr_{S_1, S_2}[\cdot]$ the error distribution before Step 3. Recall that for $i \in [q-1]$, $X_i$ is the number of deletions introduced to the interval $I_i$, which is independent of Step 3. We first observe that Step 3 does not change the relative distances among the queried indices. Therefore we have
	\begin{align*}
	& \Pr\left[ \tp{k', d_1', \dots, d_{q-1}'}\mapsto\tp{k,d_1,\dots,d_{q-1}} \right] \\
	=& \frac{1}{\floor{\delta m/4}} \cdot \sum_{e_2 = 0}^{\floor{\delta m / 4}} \Pr\left[\tp{k', d_1', \dots, d_{q-1}'}\mapsto\tp{k,d_1,\dots,d_{q-1}} \ \middle| \ e_2 \right] \\
	=& \frac{1}{\floor{\delta m/4}} \cdot \sum_{e_2 = 0}^{\floor{\delta m / 4}} \Pr_{S_1, S_2}\left[ \tp{k', d_1', \dots, d_{q-1}'}\mapsto\tp{k+e_2,d_1,\dots,d_{q-1}} \right] \\
	\le& \frac{8}{\delta m} \cdot \Pr\left[ \tp{X_1 = d_1'-d_1} \land \dots \land \tp{X_{q-1} = d_{q-1}'-d_{q-1}} \right].
	\end{align*}
	
	In the rest of the proof we will think of the error distribution as comprised of only Step 1 and 2. The chain rule of conditional probability gives
	\begin{align*}
	& \Pr\left[ \tp{X_1 = d_1'-d_1} \land \dots \land \tp{X_{q-1} = d_{q-1}'-d_{q-1}} \right] \\
	=& \Pr\left[ X_1 = d_1'-d_1 \right] \cdot \prod_{i=2}^{q-1} \Pr\left[ X_i = d_i'-d_i \mid X_1=d_1'-d_1, \dots, X_{i-1}=d_{i-1}'-d_{i-1} \right].
	\end{align*}
	
	We finish the proof with 2 claims.
	
	\begin{claim}
	$\Pr[X_1 = d_1' - d_1] \le 16/(\delta d_1')$.
	\end{claim}
	\begin{proof}[Proof of the claim]
	We are going to condition on the deletion probabilities $q_{\ell, b}$ and prove the same bound for any $q_{\ell, b} \in [0,h_{\ell}\delta]$. This clearly implies the claim. Moreover, under this conditional distribution, the deletions of individual bits in Step 1 are mutually independent.
	
	Write $X_1=X_1\itn{1}+X_1\itn{2}$ where $X_1\itn{i}$ ($i=1,2$) is the number of deletions occurred in $I_1$, introduced in Step $i$. Since Step 1 deletes each bit independently with probability at most $h\delta \le \delta/4$, Hoeffding's inequality shows that
	\begin{align*}
	\Pr\left[ X_1\itn{1} \ge \frac{1}{2}d_1' \right] \le \exp\tp{-\frac{d_1'}{2}} \le \frac{1}{d_1'},
	\end{align*}
	where the last inequality holds as long as $d_1' \ge 1$. Also notice that given $X_1\itn{1}$, $X_1\itn{2}$ follows a compound distribution $B\tp{d_1'-X_1\itn{1}-1, \*U[0,\delta/4]}$. Therefore
	\begin{align*}
	\Pr\left[ X_1 = d_1'-d_1 \right] &= \E_{X_1\itn{1}}\left[ \Pr\left[X_1\itn{2} = d_1'-d_1-X_1\itn{1} \right] \ \middle| \ X_1\itn{1} \right] \\
	&\le \E_{X_1\itn{1}}\left[ \frac{4}{\delta} \cdot \frac{1}{d_1'-X_1\itn{1}} \right] \\
	&\le \frac{4}{\delta} \cdot \frac{1}{d_1'/2} + \frac{4}{\delta} \cdot \Pr\left[ X_1\itn{1} \ge \frac{1}{2}d_1' \right] \\
	&\le \frac{16}{\delta} \cdot \frac{1}{d_1'}.
	\end{align*}
	Here the first equality uses Lemma~\ref{lem:compound-anticoncentration}.
	\end{proof}
	
	\begin{claim}
	$\forall 2\le i \le q-1$, $\Pr[ X_i = d_i'-d_i \mid \bigwedge_{j=1}^{i-1}(X_j=d_j'-d_j) ] \le 32/(\delta h_{\ell_i}d_i)$.
	\end{claim}
	\begin{proof}[Proof of the claim]
	For the $i$-th term where $2\le i \le q-1$, we recall that $\ell_i \in F_i$. Since all blocks in layer $\ell_i$ have size $s_{\ell_i} \le d_i/2 \le d_i'/2$ by the definition of $F_i$, there exists a block in layer $\ell_i$ which is completely contained in $I_i$. Suppose it is the $b$-th block and denote it by $B_i$. Note that we may also assume $|B_i| \ge d_i/4$ (if $B_i$ is the last block and $|B_i| < d_i/4$, then the second last block is also contained in $I_i$ and has size $s_{\ell_i} \ge d_i/4$). 
	
	Similar to the proof of the previous claim, we are going to condition on $\beta$ and the deletion probabilities $q_{\ell, b'}$ for all $\ell \in [L]$ and $b' \le \ceil{2m/s_{\ell}}$, except for $q_{\ell_i, b}$ which is the deletion probability of $B_i$. Proving the same bound under the conditional distribution will imply the claim.
	
	Write $X_i = X_{i, B} + X_{i, B}' + X_{i, \varnothing}$ where $X_{i, B}$ is the number of deletions introduced to $B_i$ by layer $\ell_i$,  $X_{i, B}'$ is the number of deletions introduced to $B_i$ by other sources, and $X_{i,\varnothing}$ is the number of deletions introduced to $I_i \setminus B_i$. 
	
	A crucial observation is that given $X_{i,B}'$, $X_{i,B}$ is independent of the $X_j$'s for $j \neq i$, and follows a compound distribution $B\tp{|B_i|-X_{i,B}', \*U[0,h_{\ell_i}\delta]}$. Similar to the analysis for $X_1\itn{1}$, since each bit is deleted independently with probability at most $(h+\beta)\delta \le \delta/2$ during Step 1 and 2, Hoeffding's inequality implies
	\begin{align*}
	\Pr\left[ X_{i,B}' \ge \frac{3}{4}|B_i| \right] \le \exp\tp{-\frac{|B_i|}{8}} \le \frac{4}{|B_i|},
	\end{align*}
	where the last inequality holds as long as $|B_i| \ge 1$. Therefore we have
	\begin{align*}
	& \Pr\left[ X_i = d_i'-d_i \mid X_1 = d_1'-d_1,\dots,X_{i-1} = d_{i-1}'-d_{i-1} \right] \\
	=& \E_{X_{i,B}', X_{i, \varnothing}}\left[ \Pr\left[ X_i = d_i'-d_i \mid X_1 = d_1'-d_1,\dots,X_{i-1} = d_{i-1}'-d_{i-1} \right] \ \middle| \ X_{i,B}', X_{i,\varnothing} \right] \\
	=& \E_{X_{i,B}', X_{i, \varnothing}}\left[ \Pr\left[ X_{i,B} = d_i'-d_i-X_{i,B}'-X_{i,\varnothing} \right] \ \middle| \ X_{i,B}',X_{i,\varnothing} \right] \\
	\le& \E_{X_{i,B}', X_{i,\varnothing}}\left[ \frac{1}{h_{\ell_i}\delta} \cdot \frac{1}{|B_i|-X_{i,B}'+1} \right] \\
	\le& \frac{1}{h_{\ell_i}\delta} \cdot \tp{\frac{1}{|B_i|/4} + \frac{4}{|B_i|}} = \frac{8}{h_{\ell_i}\delta} \cdot \frac{1}{|B_i|} \le \frac{32}{\delta} \cdot \frac{1}{h_{\ell_i}d_i}. 
	\end{align*}
	Here the first inequality is again due to Lemma~\ref{lem:compound-anticoncentration}. 
	\end{proof}
	
	Putting everything together, we have shown that
	\begin{align*}
	\Pr\left[ \tp{k', d_1', \dots, d_{q-1}'}\mapsto\tp{k,d_1,\dots,d_{q-1}} \right] &\le \frac{8}{\delta m} \cdot \tp{\frac{16}{\delta} \cdot \frac{1}{d_1'}} \cdot \prod_{i=2}^{q-1}\tp{\frac{32}{\delta}\cdot \frac{1}{h_{\ell_i}d_i}} \\
	&\le \frac{\tp{32/\delta}^q}{md_1} \cdot \prod_{i=2}^{q-1}\frac{1}{h_{\ell_i} d_{i}}.
	\end{align*}
\end{proof}

In the rest of the section, we instantiate $\+D(L, \*s, \*h)$ with two specific sets of parameters, which we now describe.
\subsection{An error distribution independent of the code} \label{sec:error_oblivious}
We now define an error distribution $\+D_{obl}$ which is completely independent of the coding scheme ($C \colon \set{0,1}^n \rightarrow \Sigma^m$,$\Dec$), message $x$ and codeword $C(x)$. As such lower bounds obtained from $\+D_{obl}$ will also apply in the private-key setting where the encoder and decoder share secret random coins.

We take $L_0 = \ceil{\log (2m)} \le \log m + 2$, $\*s_0 = \tp{s_1, \dots, s_{L_0}}$ and $\*h_0 = \tp{h_1, \dots, h_{L_0}}$ where
\begin{align*}
	\forall \ell \in [L_0], \quad s_{\ell} = 2^{\ell}, \quad h_{\ell} = \frac{1}{4L_0}.
\end{align*}
Let $\+D_{obl} = \+D\tp{L_0, \*s_0, \*h_0}$. Note that $\+D_{obl}$ is oblivious to the encoding/decoding scheme.

Clearly $h = 1/4$ for $\+D_{obl}$. Consider an arbitrary query $\tp{k, d_1, \dots, d_{q-1}}$. For each $i \in [q-1]$, we let $\ell_i = \ceil{\log_2 d_i}-2$. Since $\log_2 d_i - 2 \le \ceil{\log_2 d_i} - 2 \le \log_2 d_i - 1$, we have
\begin{align*}
	s_{\ell_i} = 2^{\ell_i} \ge 2^{\log_2 d_i - 2} \ge \frac{d_i}{4}, \textup{ and }s_{\ell_i} \le 2^{\log_2 d_i - 1} \le \frac{d_i}{2},
\end{align*}
which means $\ell_i \in F_i$. The corresponding $h_{\ell_i} = 1/(4L_0) \ge 1/(4(\log m + 2))$. Therefore we obtain the following corollary to Lemma~\ref{lem:distribution-generalized}.
\begin{corollary} \label{cor:distribution-obl}
Let $\tp{k', d_1', \dots, d_{q-1}'}$ be the random tuple which corresponds to the query $\tp{k, d_1, \dots, d_{q-1}}$ under error distribution $\+D_{obl}$. Then any support of $\tp{k', d_1', \dots, d_{q-1}'}$ has probability at most 
\begin{align*}
	\frac{\tp{32/\delta}^q}{md_1} \cdot \prod_{i=2}^{q-1}\frac{4(\log m + 2)}{d_i}.
\end{align*}
\end{corollary}

\subsection{An adversarial error distribution for \texorpdfstring{$q\ge 3$}{q>=3}}
\label{sec:error_adversarial}
We now define a non-oblivious error distribution $\+D_{adv,i}$ which may depend on the decoder $\Dec$. Analyzing  $\+D_{adv,i}$ allows us to derive tighter lower bounds on the codeword length $m$ for a Insdel LDC with query complexity $q$. However, because the distribution is not oblivious the stronger lower bounds derived from $\+D_{adv,i}$ no longer apply in the private-key setting.

Fix $i \in [n]$. Let $\tp{K,D_1, D_2, \dots, D_{q-1}}$ be the random variable that corresponds to queries of $\Dec\tp{\cdot, m, i}$. For $1\leq \tau \leq \lceil \log(2m) \rceil$, let $p_{\tau, i}$ be the probability that $2^{\tau-1}\leq  D_2< 2^{\tau}$. Thus, $p_{\tau, i}$ is the probability that the decoder for the $i$-th bit ($\Dec(\cdot, m, i)$) queries a tuple $(k, d_1, \dots, d_{q-1})$ such that $2^{\tau-1}\leq d_2 < 2^{\tau}$. We have $\sum_{\tau=1}^{\ceil{\log (2m)}} p_{\tau,i} = 1$. 

We take $L = 2 L_0 $ where $L_0 = \ceil{\log (2m)}$. The vectors $\*s=\tp{s_1,\dots,s_{L}}$ and $\*h=\tp{h_1,\dots,h_{L}}$ are defined as follows.
\begin{itemize}
	\item $\forall \; 1\leq \ell \leq L_0$, $s_{\ell} = 2^{\ell}$, and $h_{\ell} = 1/(8L_0)$.
	\item $\forall \; 1\leq \tau \leq L_0$, $s_{L_0+\tau} = 2^{\tau-2}$, and $h_{d+L_0} = p_{\tau,i}/8$.
\end{itemize}

We define the adversary error distribution depending on $\Dec(\cdot, m, i)$ as $\+D_{adv,i} \coloneqq \+D(L, \*s, \*h)$. For this error distribution we also have
\begin{align*}
	h = \sum_{\ell=1}^{L_0}h_{\ell} + \sum_{\tau=1}^{L_0}h_{\tau+L_0} = L_0 \cdot \frac{1}{8L_0} + \frac{1}{8} \cdot \sum_{\tau=1}^{L_0}p_{\tau,i} = \frac{1}{4}.
\end{align*}

Let $\tp{k, d_1, d_2, \dots, d_{q-1}}$ be an arbitrary query in the support of $\Dec(\cdot, m, i)$ and $1\leq \tau_0\leq t$ be the integer such that $ 2^{\tau_0-1}\leq d_2 <2^{\tau_0}$. We set $\ell_{2} = L_0 + \tau_0$. Since $s_{\ell_{2}} = 2^{\tau_0-2}$, we have $d_2/4 \leq s_{\ell_2} \leq d_2/2$. Thus, $\ell_{2}\in F_{2}$ with $h_{\ell_2} = p_{\tau_0,i}/8$. 

For $3\leq j \leq q-1 $, we set $\ell_j = \lceil \log d_j \rceil-2\in F_j$.Thus, $h_{\ell_j} = 1/(8L_0) \geq 1/(8(\log m + 2 ))$ for $3\leq j \leq q-1$. We have the following corollary to Lemma~\ref{lem:distribution-generalized}.
\begin{corollary} \label{cor:distribution-adv}
Let $\tp{k', d_1', \dots, d_{q-1}'}$ be the random tuple that corresponds to the query\\$\tp{k, d_1, \dots, d_{q-1}}$ under error distribution $\+D_{adv,i}$. Let $\tau_0$ be the integer such that $ 2^{\tau_0-1}\leq d_{2} <2^{\tau_0}$. Then any support of $\tp{k', d_1', \dots, d_{q-1}'}$ has probability at most 
\begin{align*}
\frac{\tp{32/\delta}^q}{m} \cdot \frac{8^{q-2}(\log m+2)^{q-3}}{p_{\tau_0, i}} \cdot \prod_{\ell =1}^{q-1}\frac{1}{d_\ell} .
\end{align*}
\end{corollary}

\section{Lower Bounds For Private-key Insdel LDCs}
\label{sec:LB_obl}

We will prove the second part of Theorem~\ref{thm:obliviousLBcombined} in this section, since the proof is simpler. The error distribution is going to be $\+D_{obl}$ defined in Section~\ref{sec:error_oblivious}. Because  $\+D_{obl}$ is independent of the coding scheme, message and codeword the lower bounds apply in the private-key setting. Of course the lower bounds still apply for general LDCs. However, we can derive tighter lower bounds for general LDCs using a different error distribution which may depend on the local decoder $\Dec$ --- see Section \ref{sec:LB_adv}.

Let $c=4\ln(q/\eps) \ge 2$ be the constant from Lemma~\ref{lem:distribution-generalized}. For $j_1, j_2, \dots, j_{q-1} \in [t]$ where $t = \ceil{\log_c (2m)} \le \log m + 2$, denote
\begin{align*}
P_{j_1, \dots, j_{q-1}} = [2m] \times [c^{j_1-1}, c^{j_1}) \times \dots \times [c^{j_{q-1}-1}, c^{j_{q-1}}).
\end{align*}

\begin{claim} \label{clm:fraction-generalized}
Let $\gamma = \eps/\tp{256c^2/\delta}^q$. For any $i \in [n]$, there exist $j_1, \cdots, j_{q-1} \in [t]$ such that 
\begin{align*}
\abs{P_{j_1,\dots,j_{q-1}} \cap \Goody_i} \ge \frac{\gamma \abs{P_{j_1,\dots,j_{q-1}}}}{\tp{\log m + 2}^{q-2}}.
\end{align*}
\end{claim}
\begin{proof}
Fix any $i \in [n]$. Let $D \subseteq [2m]$ be a random set of deletions generated by $\+D_{obl}$. Let $\+E$ be the event that $\abs{D \cap [m]} \le \delta m$ and $\abs{D} \le m$. By Proposition~\ref{prop:total-error-bound} and a union bound, $\+E$ happens with probability at least $1-\exp\tp{-\delta^2 m/8}-\exp\tp{-(1-\delta)^2m} \ge 1-\eps/4$ for large enough $n$ (and thus large enough $m$). Therefore by Lemma~\ref{lem:hitting}, we have
\begin{align*}
\Pr\left[ \Dec(\cdot, m, i)\textup{ hits }\Goody_i \right] &\ge \Pr\left[ \Dec(\cdot, m, i)\textup{ hits }\Goody_i \mid \+E \right] \cdot \Pr\left[ \+E \right] \\
&\ge \frac{3\eps}{2} \cdot \tp{1-\frac{\eps}{4}} \\
&\ge \frac{5\eps}{4}.
\end{align*}
Here in the case of $|D| > m$ we simply assume that $\Dec(\cdot,m,i)$ never hits $\Goody_i$. By the first item of Lemma~\ref{lem:distribution-generalized} and a union bound, for at least one query $\tp{k,d_1,\dots,d_{q-1}}$ we have
\begin{align*}
	\Pr\left[ \tp{k',d_1', \dots, d_{q-1}'} \in \tp{[2m] \times [d_1, cd_1) \times \dots \times [d_{q-1}, cd_{q-1})} \cap \Goody_i \right] \ge \frac{\eps}{4},
\end{align*}
where $\tp{k',d_1', \dots, d_{q-1}'}$ corresponds to $\tp{k,d_1,\dots,d_{q-1}}$ under $\+D_{obl}$. By Corollary~\ref{cor:distribution-obl}, we have that
\begin{align*}
	\abs{\tp{[2m] \times [d_1, cd_1) \times \dots \times [d_{q-1}, cd_{q-1})} \cap \Goody_i} \ge \frac{\eps md_1\cdots d_{q-1}}{4\tp{32/\delta}^q \cdot \tp{ 4\tp{\log m + 2}}^{q-2}}.
\end{align*}
Take $j_1, \cdots, j_{q-1} \in [t]$ such that $c^{j_{\ell}-1} \le d_{\ell} < c^{j_{\ell}}$ for all $1 \le \ell \le q-1$. Note that for each $\ell \le q-1$, $[d_{\ell}, cd_{\ell}] \subseteq [c^{j_{\ell}-1}, c^{j_{\ell}+1}) = [c^{j_{\ell}-1}, c^{j_{\ell}}) \cup [c^{j_{\ell}}, c^{j_{\ell}+1})$. This implies
\begin{align*}
	[2m] \times [d_1, cd_1] \times \dots \times [d_{q-1}, cd_{q-1}] \subseteq \bigcup_{\forall 1\le \ell \le q-1, j_{\ell}' \in \set{j_{\ell}, j_{\ell}+1}}P_{j_1', \dots, j_{q-1}'}.
\end{align*}
Therefore for some $j_1', \dots, j_{q-1}' \in [t]$ we have
\begin{align*}
\abs{P_{j_1',\dots,j_{q-1}'} \cap \Goody_i} &\ge \frac{1}{2^{q-1}} \cdot \abs{\tp{[2m] \times [d_1, cd_1] \times \dots \times [d_{q-1}, cd_{q-1}]} \cap \Goody_i} \\
&\ge \frac{1}{2^{q-1}} \cdot \frac{\eps md_1\cdots d_{q-1}}{4\tp{32/\delta}^q \cdot \tp{ 4\tp{\log m + 2}}^{q-2}} \\
&\ge \frac{\eps}{(256c^2 /\delta)^{q}} \cdot \frac{\abs{P_{j_1',\dots,j_{q-1}'}}}{\tp{\log m + 2}^{q-2}} \\
&= \frac{\gamma \abs{P_{j_1',\dots,j_{q-1}'}}}{\tp{\log m + 2}^{q-2}}.
\end{align*}
Here the last inequality is because $\abs{P_{j_1',\dots,j_{q-1}'}} \le c^{2q}\cdot 2md_1d_2\cdots d_{q-1}$.
\end{proof}

For each $i \in [n]$, we fix a tuple $\*J_i = \tp{j_1, \dots, j_{q-1}} \in [t]^{q-1}$ such that 
\begin{align*}
\abs{P_{\*J_i} \cap \Goody_i} \ge \frac{\gamma\abs{P_{\*J_i}}}{\tp{\log m + 2}^{q-2}}.
\end{align*}
Such a $\*J_i$ exists as guaranteed by Claim~\ref{clm:fraction-generalized}. Given a tuple $\*J = \tp{j_1, \dots, j_{q-1}} \in [t]^{q-1}$, define
\begin{align*}
	G_{\*J} = \set{i \in [n] \colon \*J = \*J_i}.
\end{align*}

\begin{claim} \label{clm:GJ-ub}
$\forall \*J \in [t]^{q-1}, \abs{G_{\*J}} \le q\tp{\log m + 2}^{q-2} /\tp{\gamma\tp{1 - \+H(1/2+\eps/4)}}$.
\end{claim}
\begin{proof}
By counting the number of pairs $\tp{Q, i} \in P_{\*J} \times G_{\*J}$ such that $Q \in \Goody_i$ in two ways, we have 
\begin{align*}
	\sum_{Q \in P_{\*J}}\abs{H_Q \cap G_{\*J}} = \sum_{i \in G_{\*J}}\abs{P_{\*J} \cap \Goody_i}.  
\end{align*}
On the one hand, by Proposition~\ref{prop:info-theory} we have
\begin{align*}
\sum_{Q \in P_{\*J}}\abs{H_Q \cap G_{\*J}} \le \sum_{Q \in P_{\*J}}\abs{H_Q} \le \frac{q\abs{P_{\*J}}}{1-\+H(1/2+\eps/4)}.
\end{align*}
On the other hand, by definition of $G_{\*J}$ we have
\begin{align*}
\sum_{i \in G_{\*J}}\abs{P_{\*J} \cap \Goody_i} \ge \abs{G_{\*J}} \cdot \frac{\gamma\abs{P_{\*J}}}{\tp{\log m + 2}^{q-2}}.
\end{align*}
Rearranging gives the claim.
\end{proof}

Now we are ready to prove the second part of Theorem~\ref{thm:obliviousLBcombined}.  Of course the lower bound also applies in settings where the encoding/decoding scheme do not share secret random coins, but in these settings we can establish an even stronger bound by modifying $\+D_{obl}$ to depend on the specific encoding/decoding scheme. 

\obliviousLB*
\begin{proof}[Proof of the second part]
Note that $\cup_{\*J \in [t]^{q-1}}G_{\*J} = [n]$. Therefore by Claim~\ref{clm:GJ-ub} and substituting $\gamma = \tp{256c^2/\delta}^q$ we have
\begin{align*}
	n \le \sum_{\*J \in [t]^{q-1}}\abs{G_{\*J}} \le t^{q-1} \cdot \frac{q \tp{\log m + 2}^{q-2}}{\gamma\tp{1 - \+H(1/2+\eps/4)}} \le \frac{24}{(\ln 2)^2} \cdot \frac{1}{\eps^3} \cdot \tp{\frac{512c^2}{\delta}}^q \cdot \tp{\log m + 2}^{2q-3}
\end{align*}
where in the last inequality we used Proposition~\ref{prop:entropy-estimate}, $q \le 2q-3$ for $q\ge 3$ and $q \le 2^q$ for $q \ge 1$.
Substituting $c = 4\ln(q/\eps)$ and taking $C$ to be a large enough constant, we can write
\begin{align*}
    n \le \frac{1}{\eps^3} \cdot \tp{\frac{C}{\delta} \cdot \ln^2\tp{\frac{q}{\eps}} \cdot \log m}^{2q-3}.
\end{align*}
Solving for $m$ gives
\begin{align*}
    m \ge \exp\tp{\Omega\tp{\frac{\delta}{\ln^2(q/\eps)} \cdot \tp{\eps^3 n}^{1/(2q-3)}}}.
\end{align*}
Finally, we observe that $\+D_{obl}$ is oblivious to the encoding/decoding scheme and the specific codeword. Thus, the lower bound still applies even if the encoder/decoder share secret random coins.
\end{proof}

\section{Stronger Lower Bounds For Insdel LDCs}
\label{sec:LB_adv}

In this section, we prove the first part of Theorem~\ref{thm:obliviousLBcombined}. We assume the error distribution is $\+D_{adv, i}$ introduced in section~\ref{sec:error_adversarial}. Following the notation from section~\ref{sec:error_adversarial} and section~\ref{sec:LB_obl}, let $p_{\tau, i}$ be the probability that $\Dec(\cdot, m, i)$ queries a tuple $(k, d_1, \dots, d_{q-1})$ such that $2^{\tau-1}\leq d_2 < 2^{\tau}$. We have $\sum_{\tau=1}^{\ceil{\log (2m)}} p_{\tau,i} = 1$. Take $\eta = (256/\delta)^q$, $c=4\ln(q/\eps) \ge 2$ and denote $t=\ceil{\log_c(2m)}$. For $j_1, j_2, \dots, j_{q-1} \in [t]$, denote
\begin{align*}
P_{j_1, \dots, j_{q-1}} = [2m] \times [c^{j_1-1}, c^{j_1}) \times \dots \times [c^{j_{q-1}-1}, c^{j_{q-1}}).
\end{align*}

Let $I_\tau = \{P_{j_1, \dots, j_{q-1}} \colon 2^{\tau-1} \leq c^{j_2} \leq c^2 2^{\tau}\}$ be a set of subcubes. We define 
$$\beta_{\tau,i} = \max_{P_{\*J}\in I_\tau }\frac{\abs{P_{\*J}\cap \Goody_i}}{\abs{P_{\*J}}}. $$ 
Thus, $\beta_{\tau,i}$ is the maximum fraction of good points in any subcube $P_{\*J}$ in the set $I_{\tau}$.  

\begin{claim}
\label{clm:beta_lb}
For any $i\in [n]$, we have \[\sum_{\tau=1}^t \beta_{\tau, i}\geq \frac{\eps}{8\eta(2c^2)^{q-1}(\log m+2)^{q-3} }.\]
\end{claim}

\begin{proof}

We fix $i\in [n]$. Let $Q=(k, d_1, \dots, d_{q-1})$ be an arbitrary query in the support of $\Dec(\cdot, m, i)$, and let $Q' = (k',d'_1, \dots, d'_{q-1})$ be the random tuple corresponding to $(k, d_1,\dots,d_{q-1})$ under error distribution $\+D_{adv,i}$. 

For $1\leq \ell \leq q-1$, we let $j_\ell'$ be the integer such that $c^{j_\ell'-1}\leq d_\ell < c^{j_\ell'}$. We have $[d_\ell, cd_\ell]\subseteq [c^{j_\ell'-1}, c^{j_\ell'})\cup [c^{j_\ell'}, c^{j_\ell'+1}) $. Let $U_Q$ be a set of $2^{q-1}$ tuples $\tp{j_1, \dots, j_{q-1}}$ such that $j_{\ell} \in \set{j_{\ell}', j_{\ell}'+1}$ for all $\ell \in [q-1]$ (if $j'_\ell=t$, fix $j_{\ell} = j'_\ell$). By Lemma~\ref{lem:distribution-generalized}, with probability at least $1-\eps$, we have
\[(k', d'_1, \dots, d'_{q-1})\in \bigcup_{\*J \in U} P_{\*J}. \]

Denote this event by $\+E$. We now give an upper bound of the probability that $Q'$ hits $\Goody_i$ in terms of the $\beta_{\tau,i}$'s. Let $1\leq \tau_Q \leq \ceil{ \log (2m) }$ be the integer such that $2^{\tau_Q-1}\leq d_2 < 2^{\tau_Q}$. Notice that $2^{\tau_Q-1}\leq d_2 < c^{j'_2} < c^{j'_2+1}$ and $c^{j'_2+1} \leq c^2 d_2 < c^2 2^{\tau_Q} $. We have $2^{\tau_Q-1}\leq c^{j'_2}< c^{j'_1+1}\leq c^2 2^{\tau_Q}$. Thus for any $\*J \in U_{Q}$, we have $P_{\*J}\in I_{\tau_Q}$. By our definition of $\beta_{\tau_Q, i}$, we have 
\[\beta_{\tau_Q, i}\geq \frac{\abs{P_{\*J}\cap \Goody_i}}{\abs{P_{\*J}}}.\]

By Corollary~\ref{cor:distribution-adv}, any support of $(k',d'_1, \cdots, d'_{q-1})$ has probability at most
\begin{align*}
\frac{\eta (\log m+2)^{q-3} }{md_1 \dots d_{q-1} p_{\tau_Q, i}} 
\end{align*}
for $\eta = \tp{256/\delta}^q$. 

The size of any subcube $P_{j_1,\dots,j_{q-1}}$ is bounded by $2mc^{j_1}\cdots c^{j_{q-1}} $. Since for $\*J \in U_Q$ we have $c^{j_\ell}\leq c^{j_\ell'+1}\leq c^2d_\ell$ for any $1\leq \ell \leq q-1$, we have $\abs{P_{\*J}}\leq (c^2)^{q-1}\cdot 2md_1\cdots d_{q-1}$. Thus the probability that $(k', d'_1, \dots, d'_{q-1})$ hits $\Goody_i$ can be bounded by
\begin{align*}
& \Pr\left[(k', d'_1, \dots, d'_{q-1})\text{ hits }\Goody_i \right]\\
\leq &  \Pr\left[ \overline{\+E} \right] + \Pr\left[(k', d'_1, \dots, d'_{q-1}) \in \Goody_i \cap \bigcup_{\*J \in U_Q}P_{\*J} \right] \\
\leq &  \eps + \frac{\eta (\log m+2)^{q-3} }{md_1 \cdots d_{q-1} p_{\tau_Q, i}}\cdot \beta_{\tau_Q, i} \cdot \sum_{\*J\in U_Q} \abs{P_{\*J}} \\
\leq &\eps + 2\beta_{\tau_Q, i}\cdot \frac{(2c^2)^{q-1} \eta (\log m +2)^{q-3}}{p_{\tau_Q, i}}.
\end{align*}

Denote by $\mu_i(\cdot)$ the probability distribution of the queries of $\Dec(\cdot,m,i)$. For the probability that $\Dec(\cdot, m, i)$ hits $\Goody_i$, we have
\begin{align*}
    &\Pr\left[ \Dec(\cdot, m, i)\text{ hits } \Goody_i \right] \\
    = & \sum_{Q\in \binom{[2m]}{q} }\mu_i(Q) \cdot \Pr\left[ \Dec(\cdot,m,i) \text{ hits } \Goody_i \mid \Dec(\cdot,m,i) \textup{ queries } Q \right] \\
    \leq & \sum_{Q\in \binom{[2m]}{q} } \mu_i(Q) \cdot \tp{\eps + 2 \beta_{\tau_Q, i}\cdot \frac{(2c^2)^{q-1}\eta(\log m+2)^{q-3}}{p_{\tau_Q, i}}} \\
    = & \sum_{Q\in \binom{[2m]}{q} }\mu_i(Q)\eps + \sum_{Q\in \binom{[2m]}{q} } \mu_i(Q) \cdot 2\beta_{\tau_Q, i}\cdot \frac{(2c^2)^{q-1}\eta (\log m + 2)^{q-3}}{p_{\tau_Q, i}} \\
    =& \eps + (2c^2)^{q-1}\eta(\log m+2)^{q-3} \cdot \sum_{\tau=1}^{\ceil{\log (2m)}} 2\beta_{\tau,i} \cdot \frac{1}{p_{\tau,i}} \cdot \sum_{Q \colon \tau_Q=\tau} \mu_i(Q) \\
    \le& \eps + 2(2c^2)^{q-1}\eta(\log m+2)^{q-3} \sum_{\tau=1}^{\ceil{\log(2m)}} \beta_{\tau,i}. 
\end{align*}

The last equality is due to the fact that for any $1\leq \tau \leq \ceil{\log(2m)}$, 
\[p_{\tau, i} = \sum_{Q \colon \tau_Q = \tau} \mu_i(Q).\]

Similar to the argument used in the proof of Claim~\ref{clm:fraction-generalized}, by Proposition~\ref{prop:total-error-bound} and Lemma~\ref{lem:hitting}, for large enough $n$ (and thus $m$) the probability that $\Dec(\cdot,m,i)$ hits $\Goody_i$ is at least $3\eps/2-\eps/4 = 5\eps/4$. Thus
\[ 2(2c^2)^{q-1}\eta(\log m+2)^{q-3} \sum_{\tau=1}^{\ceil{\log(2m)}} \beta_{\tau,i}\geq \eps/4.\]

\end{proof}

\begin{claim}
\label{clm:beta_ub}
$\sum_{i=1}^n\sum_{\tau=1}^{\ceil{\log(2m)}} \beta_{\tau, i}\leq 3q\log c \cdot t^{q-1}/\tp{1-\+H(1/2+\eps/4)}.$
\end{claim}
\begin{proof}
By Proposition~\ref{prop:info-theory}, for any tuple $Q\in \binom{[m]}{k}$, $Q$ is in $\Goody_i$ for at most $\frac{q}{1-\+H(1/2+\eps/4)}$ different $i$'s. Thus for any subcube $P_{\*J}$, we have
\[\sum_{i=1}^n\abs{P_{\*J}\cap\Goody_i}\leq \frac{q}{1-\+H(1/2+\eps/4)}\abs{P_{\*J}}. \]

Meanwhile, by the definition of $\beta_{\tau,i}$, we have
\[\beta_{\tau, i}=\max_{P_{\*J}\in I_\tau}\frac{\abs{P_{\*J}\cap \Goody_i}}{\abs{P_{\*J}}}\leq \sum_{P_{\*J}\in I_\tau}\frac{\abs{P_{\*J}\cap \Goody_i}}{\abs{P_{\*J}}}. \]
Combining the above two inequalities, we have 
\[\sum_{i=1}^n\beta_{\tau,i}\leq \sum_{i=1}^n \sum_{P_{\*J}\in I_\tau}\frac{\abs{P_{\*J}\cap \Goody_i}}{\abs{P_{\*J}}} \leq \frac{q}{1-\+H(1/2+\eps/4)}\abs{I_\tau}.\]

By the definition of $I_\tau$, each subcube $P_{\*J}$ belongs to at most $\ceil{\log 2 c^2}\leq 3\log c$ consecutive $I_{\tau}$'s. Notice that the total number of subcubes is bounded by $t^{q-1}$. By counting the number of subcubes, we have 
\[\sum_{\tau=1}^{\ceil{\log(2m)}} \abs{I_\tau}\leq 3\log c \cdot t^{q-1}. \]
Thus,
\[\sum_{i=1}^n\sum_{\tau=1}^{\ceil{\log(2m)}} \beta_{\tau, i}\leq  \frac{q}{1-\+H(1/2+\eps/4)}\cdot \sum_{\tau=1}^{t}\abs{I_\tau}\leq \frac{q\cdot 3\log c \cdot t^{q-1}}{1-\+H(1/2+\eps/4)}. \]

\end{proof}

Now we are ready to prove the first part of Theorem~\ref{thm:obliviousLBcombined}.

\obliviousLB*

\begin{proof}[Proof of the first part]
By Claim~\ref{clm:beta_lb}, we have 
\[\sum_{i=1}^n\sum_{\tau=1}^{\ceil{\log(2m)}} \beta_{\tau, i}\geq \frac{n\eps}{8\eta(2c^2)^{q-1}(\log m+2)^{q-3}}.\]

Combined with Claim~\ref{clm:beta_ub}, we have 
\begin{align*}
    \frac{n\eps}{8\eta(2c^2)^{q-1}(\log m+2)^{q-3}}
    & \leq \frac{q\cdot 3\log c \cdot t^{q-1}}{1-\+H(1/2+\eps/4)}.
\end{align*}

Plugging in $\eta = (256/\delta)^q$ and $c = 4\ln(q/\eps)$, and noticing that $t \le \log m + 2$, $q\le 2^q$, $3\log c \le c^2$, for some large enough constant $C$ we have 
\begin{align*}
    n &\leq \frac{1}{\eps(1-\+H(1/2+\eps/4))}\cdot \tp{\frac{C\ln^2(q/\eps)}{\delta}}^q (\log m+ 2)^{2q-4}.
\end{align*}

By Proposition~\ref{prop:entropy-estimate}, we have $1-\+H(1/2+\eps/4) = \Omega(\eps^2)$. We can rewrite the above inequality as

\[ m = \exp \tp{ \Omega\tp{ \tp{\frac{\delta}{\ln^2(q/\eps)}}^{\frac{q}{2q-4}}\cdot \tp{\eps^3 n}^{\frac{1}{2q-4}}}}.\]

Thus, for $q = 3$, we have $m = \exp(\Omega_{\delta, \eps}(\sqrt{n}))$. For $q\geq 4$, we have $\frac{q}{2q-4}\leq 1$ and  $\tp{\frac{\delta}{\ln^2(q/\eps)}}^{\frac{q}{2q-4}} = \Omega\tp{\frac{\delta}{\ln^2(q/\eps)}}$. We can write 
\[ m = \exp \tp{ \Omega\tp{ \frac{\delta}{\ln^2(q/\eps)}\cdot \tp{\eps^3 n}^{\frac{1}{2q-4}}}}.\]

\end{proof}

\bibliographystyle{alpha}
\bibliography{references}

\appendix
\section{A Note on the Definition of Insdel LDCs}\label{sec:note}
Recall the definition of Insdel LDCs from Definition \ref{def:InsdelLDC}. 

\DefInsdelLDC*

In our definition of Insdel LDCs we assume that the decoder $\Dec$ is directly given $m'$, the length of the corrupted codeword $y$. Arguably it may be more reasonable to require $\Dec$ to recover $x_i$ without a priori knowledge of the length $m'$. This question does not arise in the definition of Hamming LDCs as the length of the corrupted codeword is fixed. If we do not give the Insdel decoder the length $m'$ then $\Dec$ may query for an out of range index $j>m'$ and we would need to define how such queries are handled e.g., if $\Dec$ queries for $y[j]$ for $j > m'$ we might return $\bot$ to indicate that the query is out of range. In this case the decoder could always recover $m'$ after $O(\log m)$ queries by using binary search to find the maximum $j$ such  that $y[j] \neq \bot$.

We stress that giving the local decoder access to $m'$ can only help the decoder. If the information is not helpful the decoder can always chose to ignore the extra information $m'$. Since our focus is on proving lower bounds we chose to give $\Dec$ access to the length $m'$ which only makes the lower bounds stronger. 

Another modification of Definition \ref{def:InsdelLDC} might allow for the Insdel codewords $C(x)$ to have variable length i.e., $C: \Sigma^n \rightarrow \Sigma^{\leq m}$. Now if we require that this insdel distance between $C(x)$ and the corrupted codeword $y$ is at most $2\delta |C(x)|$ and if we additionally give $\Dec$ access to the length $m' = |y|$ of of the corrupted codeword then the encoding algorithm can ``cheat'' and use codeword length to encode $x$. For example, when $\Sigma=\{0,1\}$ and $\delta<1/6$ we could define a $(q=0, \delta, \epsilon=\frac{1}{2})$-insdel LDC $C:\{0,1\}^n \rightarrow \{0,1\}^{\leq m}$ with $m=2^{2^n}$ as follows: define a bijective mapping $\mathtt{Int}:\{0,1\}^n \rightarrow \{0,\ldots, 2^n-1\}$ in the natural way and then set
$C(x) = 1^{2^{\mathtt{Int}(x)+1}}$ i.e., $1$ repeated $2^{{\mathtt{Int}(x)+1}}$ times. If $\delta < 1/6$ then $m' = |y|$ must lie in the range \[ \frac{2}{3} 2^{{\mathtt{Int}(x)+1}} < m' < \frac{4}{3} 2^{{\mathtt{Int}(x)+1}} = \frac{2}{3} 2^{\mathtt{Int}(x)+2}  \] since we require that the insdel distance between $C(x)$ and $y$ is at most $2\delta |C(x)| < 2^{\mathtt{Int}(x)+1}/3$. This allows the local insdel decoder to recover the entire message $x$ (and any particular bit $x_i$) from $m'$ without any queries to the corrupted codeword $y$. In particular, the decoder could find the unique integer $k$ such that $\frac{2}{3} 2^{2^{k+1}} < m' < \frac{2}{3} 2^{k+2}$ and then recover $x = \mathtt{Int}^{-1}(k)$. 

Arguably the above construction ``cheats'' by encoding the message $x$ in the length of the codeword (unary) and allowing the decoder to directly learn the length of the (corrupted) codeword. If we do not allow the decoder to directly learn the length $m'$ of the corrupted message then there are no known constructions of $(q=2,\delta, \epsilon)$-insdel LDCs for any constants $\delta, \epsilon > 0$ --- for any information rate $n/m$. 

\section{Definition of Private-Key (Insdel) LDCs} \label{subsec:privkeyinsdel}
In the private-key setting the encoder $C(x; R)$ and decoder $\Dec(y,m',i; R)$ are given access to shared (secret) set of random coins $R$ which is not given to the channel $\mathcal{A}$. Fixing $x$ and $R$ we say that a corrupted codeword $y$ $\epsilon$-fools the decoder if there exists an index $i \leq n$ such that $ \Pr[\Dec(y,m',i; R) = x_i] < \frac{1}{2} + \epsilon$. In the classical setting (no shared randomness) we require that no corrupted codeword $y$ with $\ED\tp{C(x;R), y}$ has the property that it $\epsilon$-fools the decoder. In the private-key setting we relax this requirement and allow that $\epsilon$-fooling codewords $y$ exist so long as the probability the channel $\mathcal{A}$ outputs such a codeword is negligible. Formally, let $\mathtt{FoolED}(\mathcal{A}, x, R) = 1$ (resp. $\mathtt{FoolHamm}(\mathcal{A}, x, R) = 1$)  denote the event that the corrupted codeword $y=\mathcal{A}(x,C(x;R))$ output by the channel $\epsilon$-fools the decoder and $\ED\tp{C(x;R), y} \leq 2\delta m$ (resp. $\mathtt{Hamm}(C(x;R), y) \leq \delta m$). Crucially, the random coins $R$ are not known to the channel otherwise $\mathcal{A}$ could do a brute-force search to find such an $\epsilon$-fooling string. 

We say that the pair $(C,\Dec)$ is a $(q,\delta, \epsilon)$-private-key LDC for Insdel (resp. Hamming) errors if there is a negligible function $\mu(\cdot)$ such that for all channels $\mathcal{A}$ and messages $x \in \Sigma^n$ we have 
$\Pr_R[\mathtt{FoolED}(\mathcal{A}, x, R)] \leq \mu(n)$ (resp. $\Pr_R[\mathtt{FoolHamm}(\mathcal{A}, x, R)] \leq \mu(n)$). 

\paragraph{Prior Private-Key Constructions} \cite{OPS07} gives a $(q,\delta, \epsilon)$-private-key Hamming LDC with query complexity $q= \log^2 n$ and  $\epsilon > \frac{1}{2} - \mu(n)$ for a negligible function $\mu(n)$ i.e., except with negligible probability $\mu(n)$ the channel outputs a codeword $y$ such that $\forall i \leq n$ we have $\Pr[\Dec(y,m',i; R) = x_i] \geq 1-\mu(n)$. \cite{block2021private} gives a $(q,\delta, \epsilon)$-private-key Insdel LDC with query complexity $q= \log^c n$ and  $\epsilon > \frac{1}{2} - \mu(n)$ for a negligible function $\mu(n)$. The result is obtained by applying the Hamming to Insdel compiler of \cite{BlockBGKZ20} to \cite{OPS07}. 

\paragraph{$q=1$ query private-key Hamming LDC} 
We can also extend ideas from \cite{OPS07} to obtain a $(q=1, \delta, \epsilon)$-private-key Hamming LDC for suitable constants $\delta + \epsilon < \frac{1}{2}$. The length of the codeword is just $m=n \log^2 n$. In specific we define $C(x; R=(\pi, \mathtt{OTP})) = \pi\left (x^t \oplus \mathtt{OTP}\right)$. Here, $\mathtt{OTP} \in \{0,1\}^{tn}$ is a uniformly random $tn$-bit string used as a one-time-pad and $x^t \in \{0,1\}^{tn}$ denotes the string $x$ concatenated to itself $t$ times i.e., $x^1 \doteq x$ and $x^{i+1} \doteq x^{i} \circ x$. Note that for any $x$ the string $\mathtt{OTP} \oplus x^t$ is distributed uniformly at random. The random permutation $\pi: [tn] \rightarrow [tn]$  randomly shuffles the bits of $z=x^t \oplus \mathtt{OTP}$ e.g., if $z=z[1] \circ \ldots \circ z[tn]$ then $\pi(z) = z[\pi(1)] \circ \ldots \circ z[\pi(tn)]$. 

Fixing $\pi$ and an index $i \leq n$ we can define $S_{\pi,i}\doteq \{ \pi(i), \ldots, \pi(i+tn)\}$ to be the set of indices of the codeword $C(x; \pi, \mathtt{OPT})$ which correspond to $x[i]$ i.e., such that $C(x; \pi, \mathtt{OPT})[j] \oplus \mathtt{OTP}[j] = x[i]$. The decoder $\Dec(y, m', i; R=(\pi, \mathtt{OTP}))$ will randomly pick $j \in S_{\pi, i}$ and return $y[j] \oplus \mathtt{OTP}[j]$ as our guess for $x[i]$. Now a corrupted codeword $y$ $\epsilon$-fools the decoder if and only if for some $i \leq n$ at least $\left(\frac{1}{2}-\epsilon\right)$-fraction of the bits in $S_{\pi, i}$ were flipped i.e.,  $\left| \left\{j \in S_{\pi,i} : y[j] = C(x; \pi, \mathtt{OPT})[j]   \right\} \right| \leq t \left(\frac{1}{2} + \epsilon\right)$.  

Because the channel $\mathcal{A}$ does not have $\pi$ or $\mathtt{OTP}$ and can flip at most $\delta m$ bits in the codeword the expected number of bit flips in $S_{\pi, i}$ is just $\delta t< \left(\frac{1}{2}-\epsilon \right) t$. Applying concentration bounds and union bounding we have 
\[ \Pr[\exists i. ~\left| \left\{j \in S_{\pi,i} : y[j] = C(x; \pi, \mathtt{OPT})[j]   \right\} \right| \leq t \left(\frac{1}{2} + \epsilon\right)] \leq \mu(n) \] for a negligible function $\mu(n)$ whenever we set $t=\log^2 n$. Here, the randomness is taken over the selection of $\pi$ and $\mathtt{OTP}$.

\section{A Note on Hadamard Codes and Type 1 Errors}
\label{sec:hadamard_note}

In this section, we consider the following error pattern corresponding to error type 1 mentioned in the technique overview i.e., pick arbitrary $e \leq \delta m/2$ and delete the first $e$ bits from the codeword and append $e$ arbitrary bits at the end of the codeword. In particular, if $y$ is a codeword with length $m$ then the error pattern $D_e$ deletes the first $e$ bits of the codeword $y$ and then insert $e$ arbitrary bits to the end. We denote the the corrupted codeword we obtained after this error pattern by $D_e(y)$. 

This section shows that a variant of the Hadamard code allows the local decoder to recover from type 1 errors using just $q=2$ queries.

\newcommand{\Had}{\textsf{Had}}

    We prove the following theorem.
    \begin{theorem}
        There exist explicit $(2,\delta, 1/2-\delta-2^{-t})$ insdel LDCs $C \colon \set{0,1}^n \rightarrow \set{0,1}^m$ with $m=2^{nt}$, which corrects error Type \textbf{1}.
    \end{theorem}
	
	Let $\Had_{n}$ be the Hadamard encoding for message length $n$. Given a string $\*x = x_0x_2\dots x_{n-1} \in \set{0,1}^n$, we view the codeword $\Had_{n}(\*x)$ as a function $f_{\*x} \colon [2^n-1] \rightarrow \set{0,1}^n$, where
	\begin{align*}
		\forall a \in [2^n-1]\textup{ with binary representation }a = \sum_{j=0}^{n-1}a_j \cdot 2^j, \quad f_{\*x}(a) = \bigoplus_{j=0}^{n-1}a_j x_j.
	\end{align*}
	
	\textbf{Encoder} Let $t \in \mathbb{N}$ be a parameter. For a message $\*x = x_0x_1\cdots x_{n-1} \in \set{0,1}^n$, let 
	\begin{align*}
	\*x\itn{t} \coloneqq (x_00^{t-1})(x_10^{t-1})\dots (x_{n-1}0^{t-1}), 
	\end{align*}
	i.e. the string $\*x$ with $t-1$ zeros following each bit. The encoder function is given by 
	\begin{align*}
	\Enc(\*x) = \Had_{tn}(\*x\itn{t}).
	\end{align*}
	Therefore the codeword length is $m = 2^{tn}$. Note that $x_i$ has index $ti$ in $\*x\itn{t}$, we can recover $x_i = f_{\*x\itn{t}}(2^{ti})$.
	
	\textbf{Decoder} Consider the following decoder $\Dec$. Let $f \colon [2^{tn}-1] \rightarrow \set{0,1}$ be the received string. To decode $x_i$, $\Dec$ picks a random $a \in [2^{tn}-1]$ such that $a + 2^{ti} \le 2^{tn}-1$ and outputs $f(a) \oplus f(a+2^{ti})$. 
	
	\textbf{Analysis}
	For integers $a, b \in \mathbb{N}$ with binary representations $a_j$ and $b_j$ we write $a \le_2 b$ if $a_j \le b_j$ for all $j$. We say $a \le 2^{tn}-2^{ti}-1$ is \emph{bad} for $i$ if $2^{ti}(2^t-1) \le_2 a$, i.e. $a_j = 1$ for all $ti \le j \le t(i+1)-1$. Otherwise we say $a$ is \emph{good} for $i$.
	
	The success rate of the decoder is implied by the following two claims.
	\begin{claim}
	If $a$ is good for $i$, then for any $\*x \in \set{0,1}^n$, $f_{\*x\itn{t}}(a) \oplus f_{\*x\itn{t}}(a + 2^{ti}) = x_i$.
	\end{claim}
	\begin{proof}
	Let $a_j$ and $a_j'$ be the binary representations of $a$ and $a+2^{ti}$, respectively. Note that $a_j = a_j'$ for $j < ti$, and $a_{ti} \neq a_{ti}'$. If $a$ is good for $i$, we must also have $a_j = a_j'$ for $j \ge t(i+1)$, since otherwise there exists $j_0 \ge t(i+1)$ such that $a_{j_0}' = 1$ and  $a_{j_0} = 0$, and then
	\begin{align*}
		2^{ti} = (a+2^{ti})-a = \sum_{j=ti}^{j_0}(a_j' - a_j)2^j > 2^{j_0} - \sum_{j=ti}^{j_0-1}2^j = 2^{ti},
	\end{align*}
	which is a contradiction. The inequality is strict because $a_j = 0$ for at least one $ti \le j \le j_0-1$.
	
	Finally, note that $x\itn{t}_j = 0$ for $ti < j < t(i+1)$. We thus have
	\begin{align*}
		f_{\*x\itn{t}}(a) \oplus f_{\*x\itn{t}}(a + 2^{ti}) = \bigoplus_{j=ti}^{t(i+1)-1}a_j x\itn{t}_j \oplus \bigoplus_{j=ti}^{t(i+1)-1}a_j' x\itn{t}_j = a_{ti} x\itn{t}_{ti} \oplus a_{ti}' x\itn{t}_{ti} = \tp{a_{ti} \oplus a_{ti}'}x_i = x_i.
	\end{align*}
	\end{proof}

	\begin{claim}
	For any $0\le i\le n-1$, $\Pr_{a \in [2^{tn}-2^{ti}-1]}\left[ a\textup{ is good for }i \right] \ge 1-2^{-t}$.
	\end{claim}
	\begin{proof}
	Note that for $a \ge 2^{t(i+1)}$, $2^{ti}(2^t-1) \le_2 a$ if and only if $2^{ti}(2^t-1) \le_2 a - 2^{t(i+1)}$. Therefore we only need to show $\Pr_{a \in [2^{t(i+1)}-1]}\left[ a\textup{ is good for }i \right] \ge 1-2^{-t}$. Clearly $a$ is good for $i$ when $a < 2^{ti}(2^t-1)$, and hence
	\begin{align*}
		\Pr_{a \in [2^{tn}-2^{ti}-1]}\left[ a\textup{ is good for }i \right] \ge \frac{2^{ti}(2^t-1)}{2^{t(i+1)}} = \frac{2^t-1}{2^t} = 1 - 2^{-t}.
	\end{align*}
	\end{proof}

	If the distances among the indices are preserved, each deletion at the beginning of the codeword reduces the number of good indices by at most 1. So after $\delta m$ deletions we are left with at least
	\begin{align*}
		(1-2^{-t})(2^{tn}-2^{ti}) - \delta \cdot 2^{tn} = (1-\delta)2^{tn} - 2^{t(n-1)} - 2^{ti} \ge (1-\delta - 2^{-t+1})2^{tn}
	\end{align*}
	good $a$'s for $i$. Therefore it is a $(2,\delta,1/2-\delta-2^{-t+1})$ LDC for such errors.

\section{Proof of Theorem~\ref{thm:LCCtoLDCreduction}}\label{sec:LCCtoLDC}

The following theorem says that from any linear (resp. affine) insdel LCC, we can obtain a linear (resp. affine) insdel LDC that has the same parameter. The theorem is essentially Lemma 2.3 from ~\cite{Yekhanin12}. 

\begin{theorem}\label{thm:LCCtoLDCreduction_linear}
Let $\mathbb{F}$ be a finite field. Suppose $C\subseteq \mathbb{F}^m$ is a linear (resp. affine) $(q, \delta, \eps)$-insdel LCC, then there exists a linear (resp. affine) $(q, \delta, \eps)$-insdel LDC $C'$ encoding messages of length $\mathsf{dim}(C)$ to codewords of length $m$. The same holds for the insdel LCCs and insdel LDCs in the private-key setting.
\end{theorem}

\begin{proof}[Proof of Theorem~\ref{thm:LCCtoLDCreduction_linear}]
For any linear code $C\subseteq \mathbb{F}^m$, it encodes a message $\mathbf{x}\in \mathbb{F}^n$ to a codeword $\mathbf{y}\in \mathbb{F}^m$ through encoding function $\mathbf{y} = \mathbf{x} \cdot G$ with generating matrix $G\in \mathbb{F}^{n\times m}$. Let $I\subseteq [m]$ be a set of $\mathsf{dim}(C)$ information coordinates of $C$ (i.e. a set of coordinates whose value uniquely determines an element in $C$). For $\mathbf{y}\in C$, let $\mathbf{y}|_I\in \mathbb{F}^n$ be the restriction of $\mathbf{y}$ to the coordinates in $I$. We can find another generator matrix $G'$ such that for any message $\mathbf{x}\in \mathbb{F}^n$, $\mathbf y = \mathbf x\cdot G'\in C$ and $\mathbf y|_I = \mathbf x$. It is easy to verify that the locally correctability of $C$ implies the locally decodability of $C'$.

If $C$ is an affine code, the encoding function becomes $\mathbf{y} = \mathbf{x} \cdot G + \mathbf{b}$ for some $\mathbf{b}\in \mathbb{F}^m$. Again, let $I\subseteq [m]$ be a set of $\mathsf{dim}(C)$ information coordinates of $C$. Similarly, we can pick generator matrix $G'$ such that for any message $\mathbf{x}\in \mathbb{F}^n$, $\mathbf y = \mathbf x\cdot G' + \mathbf b\in C$ and $\mathbf y|_I - \mathbf b|_I= \mathbf x$. The locally correctability of $C$ implies the locally decodability of $C'$. 

We note the above argument holds for insdel LCCs and insdel LDCs in the private-key setting. 

\end{proof}

The proof of Theorem~\ref{thm:LCCtoLDCreduction} uses the same reduction introduced by \cite{bhattacharyya2016lower}. We first introduce two lemmas. The following lemma says that insdel LCCs must have large Hamming distance.

\begin{lemma}
\label{lem:lcc_hamming}
If $C\in \Sigma^m$ is a $(q, \delta, \eps)$-LCC, then for any two codewords $c, c'\in C$, the Hamming distance between $c$ and $c'$ is larger than $2\delta m$. 
\end{lemma}

\begin{proof}
Assume there are two codewords $c, c'\in C$ such that $c\neq c'$ and the Hamming distance between $c$ and $c'$ is at most $2\delta m$. Then we can conclude that $\ED(c, c') \le 4\delta m$. This is because we can transform $c $ into $c'$ with less then $2\delta m$ symbol substitutions. Meanwhile, each symbol substitution can be replaced by an insertion and a deletion. Thus it takes at most $4\delta m$ insertions/deletions to transform $c$ into $c'$. 

Since $\ED(c, c') \le 4\delta m$, there must exists a $y\in \Sigma^{m'}$ such that $\ED(c, y)\le 2\delta m$ and $\ED(c', y) \le 2\delta m$. Let $i$ be one of the positions that $c$ and $c'$ differs. By the definition of insdel LCCs, there is a probabilistic algorithm $\Dec$ such that $\Pr\left[ \Dec(y, m', i) = c_i \right] \ge \frac{1}{2} + \eps$ and $\Pr\left[ \Dec(y, m', i) = c'_i \right] \ge \frac{1}{2} + \eps$. This is impossible because $c_i\neq c'_i$. Thus, the Hamming distance between $c$ and $c'$ must be larger than $2\delta m$.

\end{proof}

The reduction needs the following notion of VC-dimension. 

\begin{definition}[VC-dimension]
Let $A\subseteq \{0,1\}^n$, then the VC-dimension of $A$, denoted by $\mathsf{vc}(A)$ is the cardinality of the largest set $I\subseteq [n]$ which is shattered by $A$. That is, the restriction of $A$ to $I$, $A|_I = \{0,1\}^{|I|}$.
\end{definition}

The following lemma due to Dudley~\cite{dudley1978central} is the key to the reduction. 

\begin{lemma}[Theorem 14.12 in \cite{ledoux2013probability}]
\label{lem:vc_dimmension}
Let $A\subseteq \{0,1\}^n$ such that for every distinct $x, y\in A$, $ \| x-y\|_{\ell_2} \ge \eps \sqrt{n}$. Then, 
\[\mathsf{vc}(A) = \Omega\tp{\frac{\log(|A|)}{\log(2/\eps)}}. \]
\end{lemma}

\begin{proof}[Proof of Theorem~\ref{thm:LCCtoLDCreduction}]
Assume $C: \{0,1\}^k \rightarrow \Sigma^m$ is a $(q,\delta, \eps)$-insdel LCC. Without loss of generality, we can assume $\Sigma = \{0,1\}^s$. Consider an error correcting code $C^0:
\{0,1\}^s \rightarrow \{0,1\}^t$ with constant distance $\delta_0$, i.e. any two codewords in $C^0$ has Hamming distance $\delta_0 t$. Let $C^1:\{0,1\}^k\rightarrow \{0,1\}^{m t}$ be the concatenation of code $C$ and $C^0$. That is, every codeword in $C^1$ is obtained by first encoding the message with code $C$ and then encoding each symbol in $\Sigma$ with code $C^0$. By Lemma~\ref{lem:lcc_hamming}, any two codewords in $C^1$ must have Hamming distance at least $2\delta\delta_0 m t$. Thus, let $\eps= \sqrt{2\delta \delta_0}$, for any $c, c' \in C^1$, we have 
\[\| c - c'\|_{\ell_2}\geq \eps\sqrt{mt}.\]
By Lemma~\ref{lem:vc_dimmension}, the VC-dimension of $C^1$ ( $\mathsf{vc}(C^1)$) is  \[\Omega\tp{\frac{\log(|C^1|)}{\log(2/\eps)}} = \Omega\tp{\frac{\log(|C|)}{\log(2/\eps)}}  = \Omega\tp{\frac{k}{\log(1/\delta)}} \]

Let $k' = \mathsf{vc}(C^1)$. By the definition of VC-dimension, there exists a set $I\subseteq [mt]$ of size $k'$ that shatters $C^1$, i.e. $C^1|_{I} = \{0,1\}^{|I|}$. 

We can build a $(q, \delta, \eps)$-insdel LDC $C':\{0,1\}^{k'} \rightarrow \Sigma^m$ as follows. For any message $x\in \{0,1\}^{k'}$, let $C'(x) = z \in C$ such that $C^0(z)|_I = x$. Here, by $C^0(z)$, we mean encoding each symbol of $z$ with code $C^0$. By the property of set $I$, we know such a codeword $z \in C$ must exist. If there are more than one $z\in C$ satisfying $C^0(z)|_I = x$, we pick one of them arbitrarily. Assuming $I = \{t_1, t_2, \dots, t_{k'}\}$, for any message $x \in \{0,1\}^{k'}$ and any $i \in [k']$, we have $ x_i =  \tp{C^0(C'(x))}_{t_i}$. 

We now show that $C'$ is indeed a $(q, \delta, \eps)$-insdel LDC. To decode a message bit $x_i$, we only need to look at the block of $\tp{\{0,1\}^{t}}^m$ that contains $t_i$, say it is the $j$-th block.
Assume we are given a $y\in \Sigma^*$ with $\ED\tp{y, C'(x)} \le 2\delta m$, and notice that $C'(x)$ is also a codeword of $C$. The decoding algorithm for $C$ can recover $C'(x)_j$ with probability of at least $1/2 + \eps$, using at most $q$ queries to $y$. Applying $C^0$ to $C'(x)_j$ will give us $x_i$. Thus, $x_i$ can be decoded with probability at least $1/2 + \eps$ using at most $q$ queries to $y$.

\end{proof}

\end{document}